%% file: paper.tex
\documentclass[a4paper,UKenglish,cleveref, autoref, thm-restate]{lipics-v2021}
\nolinenumbers
\allowdisplaybreaks[2]

\pdfoutput=1 %
\hideLIPIcs  %

\usepackage{bcprules}
\usepackage{amsmath}
\usepackage{amssymb}
\usepackage{amsthm}
\usepackage{mathtools}
\usepackage{stmaryrd}
\usepackage{multicol}
\usepackage{listings}
\usepackage{booktabs}
\usepackage[dvipsnames]{xcolor}
\usepackage{tikz}
\usetikzlibrary{arrows.meta,calc,positioning}

\lstdefinelanguage{Rust}{
  keywords={
    as, break, const, continue, crate, else, enum, extern,
    false, fn, for, if, impl, in, let, loop, match, mod, move,
    mut, pub, ref, return, self, Self, static, struct, super,
    trait, true, type, unsafe, use, where, while, async, await,
    dyn
  },
  sensitive=true,
  morecomment=[l]{//},
  morecomment=[s]{/*}{*/},
  morestring=[b]"
}

\input{local}

\title{Prophecy-Based Automated Verification of Message-Passing Programs} %

\author{Takashi Nagatomi}{The University of Tokyo, Japan}{}{}{}
\author{Musashi Katsura}{The University of Tokyo, Japan}{}{}{}
\author{Naoki Kobayashi}{The University of Tokyo, Japan}{}{https://orcid.org/0000-0002-0537-0604}{JSPS KAKENHI Grant Numbers JP20H05703, JP26H02486}
\author{Yusuke Matsushita}{Kyoto University, Japan \and MPI-SWS, Germany}{}{https://orcid.org/0000-0002-5208-3106}{Hakubi Project at Kyoto University, JSPS KAKENHI Grant Number JP24KJ0133}
\author{Ken Sakayori}{The University of Tokyo, Japan}{}{https://orcid.org/0000-0003-3238-9279}{JSPS KAKENHI Grant Numbers JP24K20731, JP26H02486}

\authorrunning{T.\ Nagatomi, M.\ Katsura, N.\ Kobayashi, Y.\ Matsushita and K.\ Sakayori } %

\Copyright{Jane Open Access and Joan R. Public} %



\EventEditors{Ana Sokolova and Patrick Totzke}
\EventNoEds{2}
\EventLongTitle{37th International Conference on Concurrency Theory (CONCUR 2026)}
\EventShortTitle{CONCUR 2026}
\EventAcronym{CONCUR}
\EventYear{2026}
\EventDate{September 1--4, 2026}
\EventLocation{Liverpool, UK}
\EventLogo{}
\SeriesVolume{391}
\ArticleNo{5}

\begin{document}

\maketitle

\begin{abstract}
  We propose a fully automated method for verifying functional correctness of message-passing concurrent programs by reducing verification problems to constrained Horn clause (CHC) solving. Inspired by RustHorn’s prophecy-based technique, we represent each sender channel by a list of values to be sent over the channel in the future, which enables modular encoding of sender and receiver threads in CHCs. To capture causal dependencies between different channels, we further attach timestamps to messages. We prove that the resulting reduction is sound and complete: a program is free from assertion failures if and only if the corresponding system of CHCs is satisfiable. We have also implemented a prototype verifier for Rust-like programs and experimentally confirmed the effectiveness of the approach.
\end{abstract}

\input{sections/Introduction}
\input{sections/TargetLanguage3}

\input{sections/Translation3}

\input{sections/ImplementationAndEvaluation.tex}

\input{sections/Discussion.tex}
\input{sections/Related.tex}

\input{sections/Conc.tex}
\bibliography{reference}

\clearpage
\appendix
\input{sections/TranslationAppx}

\input{sections/ProofOfProperties}

\input{sections/EvaluationAppx.tex}

\end{document}

%% file: sections/Introduction.tex
\section{Introduction}
\label{sec:intro}
We propose a fully automated method for verifying functional
correctness of message-passing concurrent programs. Fully
automated verification of such programs is challenging because their
behavior is non-deterministic and depends on complex interactions through message
exchanges. Although type-based approaches to automated analysis and
verification of message-passing programs have been proposed~\cite{kobayashi2003type,huttel2016foundations}, they are
usually conservative and mainly focus on abstract properties such as
deadlock freedom, information-flow security, and protocol
conformance. In contrast, automatically verifying functional
correctness properties such as ``the number of messages received is
\(n\)'' and ``the result of the whole computation is \(2n\)'' (where
\(n\) may be a program parameter) remains challenging,
despite recent progress such as the refinement-session-type-based approach of
Toby and Ankush~\cite{10.1007/978-3-032-22723-2_11}.

Inspired by RustHorn’s prophecy-based technique~\cite{matsushita2021rusthorn} for reducing
verification of sequential programs with mutable references to CHC
solving, that is, satisfiability checking of constrained Horn clauses,
we propose a prophecy-based reduction from verification of
message-passing concurrent programs to CHC solving~\cite{Bjorner15,Angelis}.

\begin{figure}[tbp]
\begin{verbatim}
fn msg_count(n: usize) {
    let (s, r) = channel();
    for i in 0..n { // spawn n senders
        let s_clone = s.clone();
        thread::spawn(move || {
            s_clone.send(1).unwrap(); // each thread sends 1 to the channel
        });
    }
    let mut c = 0; // number of messages received
    loop {  // repeatedly receive and count messages
        let _v = r.recv().unwrap();
        c += 1;  // increments the message counter
        assert!(c <= n); // error if more than n messages are received
    }
}
\end{verbatim}
\caption{A Rust-style Message-Passing Concurrent Program}
\label{fig:rust-program}
\end{figure}

To illustrate the
idea, consider the Rust-style concurrent program in Figure~\ref{fig:rust-program}.
The function \texttt{msg\_count} takes a natural number \(n\) and
creates a sender channel \texttt{s} and a corresponding receiver channel \texttt{r}.
It then spawns \(n\) threads, each of which sends \(1\) through the channel.
It next repeatedly receives
messages, counts them, and asserts that the number of messages
received so far is at most \(n\).  Suppose we wish to verify that the
assertion never fails.
We convert this problem to the satisfiability problem for the following constrained
Horn clauses (CHCs).
\begin{align*}
 & \bot \Leftarrow \entry(\TRUE, n)\land n\ge 0 .\\
  &  \entry(b, n) \Leftarrow \forloop(b, 0, n, l, l).\\
  &\forloop(b, n, n, \NIL, r) \Leftarrow \rloop(b, 0, n, r).\\
 & \forloop( b_1\lor b_2, i, n, s, r) \\
  &\qquad \Leftarrow i<n \land \MERGE(s_1,s_2,s)\land \send(b_2, s_2)\land   \forloop(b_1, i+1, n, s_1, r).\\
  & \send(\FALSE, s) \Leftarrow s=[1]. \\ %
  & \rloop(b, c, n,r) \Leftarrow r=v::r' \land c+1\le n \land \rloop(b, c+1, n,r').\\
  & \rloop(\TRUE, c, n,r) \Leftarrow r=v::r' \land c+1> n.
\end{align*}
The system of CHCs above is constructed so that it is satisfiable, i.e.,
there exists an interpretation of the predicates that makes all the clauses valid,
if and only if, for every non-negative integer \(n\), \(\texttt{msg\_count}(n)\)
never causes an assertion failure.
Thus, the verification problem can be reduced to CHC satisfiability and handled
by a CHC solver~\cite{DBLP:journals/fmsd/KomuravelliGC16,hojjat2018eldarica,DBLP:conf/aplas/Champion0S18,katsura2025automated}.

Intuitively, the predicate \(\entry(b, n)\) means that a call to
\(\texttt{msg\_count}(n)\) may terminate, with \(b=\TRUE\) indicating that
the call may terminate with an assertion failure.
Thus, the first clause states that a call to \(\texttt{msg\_count}(n)\) never
terminates with an assertion failure.
The predicate \(\forloop(b, i,n,s,r)\) means that, from the program state at the beginning
of the \texttt{for}-loop with loop index \(i\), the remaining execution may terminate,
with \(b\) indicating the error status.
Similarly, the predicate \(\rloop(b, c,n,r)\) means that, from the program state at the
beginning of the second loop with counter value \(c\), the remaining execution may terminate, with \(b\) indicating the error status.

The main issue is how to represent channels in CHCs.
Inspired by RustHorn's prophecy-based technique, we represent a sender channel
\(s\) as a list of values to be sent \emph{in the future}, and a receiver channel
\(r\) as a list of values that may be received \emph{in the future}.
Thus, we model the behavior of the program as follows.
When creating a channel, the program \emph{guesses} a list \(l\) of values to be sent
through the channel.
The second clause reflects this behavior: it says that if the execution from the beginning of the \texttt{for}-loop 
may terminate with error flag \(b\) by exchanging a list \(l\) of values through the channel,
then \(\texttt{msg\_count}(n)\) may terminate with error flag \(b\).
Cloning of the sender channel is represented by the list-merge relation
\(\MERGE(s_1,s_2,s)\) in the fourth clause, which means that the list \(s\) is
obtained by merging \(s_1\) and \(s_2\) in some order.
This models the creation of a clone of channel \(s\): the two resulting sender channels
are represented by \(s_1\) and \(s_2\), and the total list of values to be sent is obtained
by merging them.
The fifth clause, \(\send(\FALSE, s) \Leftarrow s=[1] \), models the send operation.
The equation \(s=[1]\) ensures that the prophecy \(s\), representing the list of values
to be sent, matches the actual list of messages being sent, namely \([1]\).
The last two clauses, in contrast, model the receive operation in the program.
Receiving a value \(v\) from a receiver channel \(r\) is modeled by the equation
\(r=v::r'\), which assumes that \(v\) is at the head of the receiver queue.
Note that there is no clause for the case where \(r\) is an empty list; this means that
no message is available, and hence the execution is blocked.

The crux of the prophecy-based representation of channels is that it enables
modular modeling of sender and receiver threads: it suffices that the sender
and receiver agree, when a channel is created (as in the second clause above),
 on the values to be sent, and afterwards, the sender locally checks at each
send operation that the prophecy is correct, while the receiver locally assumes
that messages arrive as described by the prophecy.
One may be tempted to model a channel by the current contents of its message queue,
rather than by the list of values to be sent in the future.
In that case, however, the message queue would have to be treated as global state:
since it is updated by both send and receive operations, sender and receiver threads
could no longer be modeled as separate predicates in the resulting CHCs.

The idea above is sufficient to ensure soundness: if the corresponding system of CHCs
is satisfiable, then the program does not cause any assertion failure.
However, it is not sufficient for completeness.
This is because the prophecy-based representation of channels does not capture
the causal dependency between different channels. To address this issue,
we attach a \emph{timestamp} to each value to be sent and represent
a channel as a list of pairs consisting of a timestamp and a value: for example,
\([(1, v_1); (3, v_2)]\) means that the value \(v_1\) is to be sent at time \(1\),
and \(v_2\) is to be sent at time \(3\). We show that the reduction to CHCs
extended with timestamps is sound and complete:
a program does not cause any assertion failure if and only if
the corresponding system of CHCs is satisfiable.

The contributions of this paper are summarized as follows.
\begin{itemize}
\item Prophecy- and timestamp-based reduction from the safety verification of
  message-passing concurrent programs to CHC solving.
  We prove that the reduction is sound and complete.
  We also discuss a few variations of the reduction, which capture
  subtle variations of message-passing semantics.
\item Implementation and experiments.
  We have implemented a prototype verification tool for Rust-like programs
  based on the above idea, and conducted experiments to evaluate its effectiveness.
  While the experiments confirm the effectiveness of our method,
they also suggest the need for further improvement of the backend CHC solvers,
which would also benefit the community of CHC-based verification.
\end{itemize}

The rest of this paper is structured as follows.
Section~\ref{sec:target-language} introduces a language with message-passing concurrency.
Section~\ref{sec:Translation} formalizes our reduction from the safety verification of message-passing concurrent programs to CHC solving, and Section~\ref{sec:implementation-evaluation} reports experimental results.
Section~\ref{sec:discussion} discusses the flexibility of our approach under variations of asynchronous message-passing semantics.
Section~\ref{sec:rel} discusses related work, and Section~\ref{sec:conc} concludes.

%% file: sections/TargetLanguage3.tex
\section{Target Language}
\label{sec:target-language}
This section defines the target language, a
message-passing concurrent language
augmented with first-order recursive functions.
The communication channels obey the multi-producer, single-consumer (MPSC) principle as in Rust.

\textbf{Notation.}
We write $\seq{a}$ to denote a sequence $a_1, \ldots, a_n$.
This notation is also extended to binary relations in a pointwise manner.
For example, we  write $\seq{x}:\seq{\TYPE}$ for a sequence of type bindings $x_1:\TYPE_1, \ldots, x_n:\TYPE_n$.

\subsection{Syntax and Types}
\label{sec:target-syntax}

\subparagraph*{Syntax}
In our target language,
a program consists of function definitions, each of which
 only performs a single ``instruction''.
Control flow is encoded as function calls, and every function ends with a call to a continuation function.
This design is chosen to simplify the translation to CHCs.

The syntax is given as follows:
\begin{align*}
  \mbox{ (statements) } M := &\  \FAIL \mid () \mid f(\seq{a}) \mid \ifstmt{a}{f_1(\seq{a})}{f_2(\seq{a})} \\
    \mid &\ \sendstmt{a}{x}{f(\seq{a})} \mid \ \recvstmt{x}{y}{f(\seq{a})} \mid \newstmt{x}{y}{f(\seq{a})}  \\
    \mid &\ \dupstmt{x}{y}{f(\seq{a})} \mid \spawnstmt{f_1(\seq{a_1})}{f_2(\seq{a_2})} \\
    \mbox{ (arithmetic exp.) }a  ::= &\ n  \mid x \mid a_1\ \OP\ a_2 \\
  \mbox{ (function defs.) }D  ::= &\ \{\,f_1(\seq{x}_1\COL{}\seq{b}_1)\ =\ M_1,\ldots,f_n(\seq{x}_n\COL{}\seq{b}_n)\ =\ M_n\,\} \\
  \mbox{ (program) } P ::= &\  (D, f(\seq{a}))\\
  \mbox{ (argument types) }\bt  ::=&\ \INT \mid \ReceiverType \mid \SenderType \\
  \mbox{ (types) }  \TYPE  ::=& \ \bt \mid \ut \mid (\bt,\ldots,\bt) \rightarrow \ut
\end{align*}
The definition of \emph{arithmetic expressions} and the constructs in the first line should be mostly self-explanatory.
We briefly note a few points.
Here \( \OP \) is a meta-variable for binary integer operators.
Comparison and logical operators, such as \( = \) or \( \land \), are also regarded as integer operators by regarding \( \{0, 1 \}\) as Boolean values, and we may write \( \TRUE \) and \( \FALSE \) in place of \( 1 \) and \( 0 \), respectively.
The statement \( \FAIL \) represents abnormal program termination.
The second line shows basic channel-related operations: sending
the value of \(a\) via the channel \( x \), receiving an integer from the channel \( y \) and channel creation.
Channel creation \( \newstmt{x}{y}{f(\seq{a})} \) creates the two endpoints of the channel, where \( x \) is the sender channel and \( y \) is the corresponding
receiver channel, and proceeds by calling \( f \). %
The operator \( \mathtt{clone} \) duplicates (i.e.~creates an alias of) a sender channel.
The statement \( \spawnstmt{f_1(\seq{a_1})}{f_2(\seq{a_2})} \) creates a new thread to execute \( f_1(\seq{a_1}) \), while the current thread continues with \( f_2(\seq{a_2}) \).
A \emph{program} is a pair of function definitions and the main function.
The type \( \ut \) denotes the unit type; note that the return type of a function is restricted to the unit type.
The types \( \ReceiverType \) and \( \SenderType \) represent types of receiver and sender
channels, respectively.

\bigskip
\begin{example}
  \label{ex:running-ex}
  \newcommand*{\mcnt}[1]{\mathit{msg\_count}(#1)}
  \newcommand{\fnsymb}[1]{f_{#1}}
  \newcommand*{\fn}[2]{\fnsymb{#1}(#2)}
  \newcommand*{\for}[1]{\mathit{for}(#1)}
  \newcommand*{\loopfn}[1]{\mathit{loop}(#1)}
  \newcommand*{\assert}[1]{\mathit{assert}(#1)}
  \newcommand*{\fail}[1]{\mathit{fail}(#1)}
  The program in the introduction can be written in our target language as \( (D, \mcnt n)  \), where the function definitions in \( D \) are given as follows:
\begin{align*}
  &\mcnt {n \COL{} \INT} = %
  \newstmt s r{\for{0, n, s, r}}\\
  &\for{i \COL{} \INT, n \COL{}\INT, s \COL{} \SenderType, r \COL{} \ReceiverType} = \ifstmt {i < n}{\fn 4 {i, n, s, r}}{\loopfn {0, n, r}}  \\
  &\fn 4 {i \COL{} \INT, n \COL{}\INT, s \COL{} \SenderType, r \COL{} \ReceiverType} = \dupstmt {s'} {s} {\fn 5 {i, n, s, r, s'}} \\
  &\fn 5 {i \COL{} \INT, n \COL{}\INT, s \COL{} \SenderType, r \COL{} \ReceiverType, s' \COL{} \SenderType} = \spawnstmt {\fn 6 {s'}}{\for{i + 1, n, s, r}}\\
  &\loopfn{c \COL{} \INT, n \COL{} \INT, r \COL{} \ReceiverType} = \recvstmt {\_v} r {\assert {c + 1, n, r}}\\
  &\assert{c\COL{} \INT, n \COL{} \INT, r \COL{} \ReceiverType} = \ \ifstmt{c \le n}{\loopfn {c, n, r}}{\fail{c, n, r}} \\
  &\fail{c\COL{} \INT, n \COL{} \INT, r \COL{} \ReceiverType} = \FAIL \\
  &\fn 6 {s \COL{} \SenderType} = \sendstmt 1 s {()} \\
\end{align*}
We just prepare a function for each program point whose arguments are variables that are alive at that point.
  Roughly speaking, \( \fnsymb i\) corresponds to the \( i \)-th line of \autoref{fig:rust-program}. We note that the assertions can be encoded using \(\FAIL\).
\qed
\end{example}

\subparagraph*{Typing}
We consider only well-typed programs, as defined by 
the typing rules in \autoref{fig:target-typing-rules}.
A \emph{type environment}, written \( \tenv = x_1 : \TYPE_1, \ldots, x_n : \TYPE_n \), is a finite map from variables to types.
We treat channel types as linear types, and
say that \( \tenv \) is non-linear if it contains no bindings of the form \( x\COL\ReceiverType\) or \( x\COL\SenderType\).
We write \( \tenv_1 + \tenv_2 \) for the union of the two environments, which is only defined if the linear parts of the domains of $\tenv_1$ and $\tenv_2$ are disjoint.
The typing rules are fairly standard, except for the linear (or affine, more precisely)
treatment of channel types:  for example, the statement
\(\spawnstmt{f_1(x)}{f_2(x)}\) is not allowed if \(x\) has a channel type.
Weakening of a variable with a channel type is only allowed in the rules for \( () \) and \( \FAIL \).
\begin{figure}[tbp]
\begin{multicols}{2}
\infrule
  {\NONLIN{\tenv}}
  {\tenv, x\COL{} \bt \vdash x : \bt}
\infrule
  {\NONLIN{\tenv}}
  {\tenv \vdash n : \INT}
\infrule
  {\tenv \vdash a_1 : \INT \andalso \tenv \vdash a_2 : \INT}
  {\tenv \vdash a_1\ \OP\ a_2 : \INT}
\infrule{}{\tenv \vdash \FAIL : \ut}
\infrule{}{\tenv \vdash () : \ut}
\infrule
  {\tenv_0(f) = (\bt_1,\ldots,\bt_n) \rightarrow \ut \andalso \NONLIN{\tenv_0} \\ \tenv_i \vdash a_i : \bt_i\ \mbox {(for each i $\in \{\,1,\ldots,n\,\}$)}}
  {\tenv_0 + \tenv_1 + \ldots + \tenv_n \vdash f(a_1,\ldots,a_n) : \ut}
\infrule
  {\tenv_0 \vdash a : \INT\\ \tenv \vdash f_i(\seq{a}_i) : \ut \mbox{ (for $i = 1, 2$) }}
  {\tenv_0 + \tenv \vdash \ifstmt{a}{f_1(\seq{a}_1)}{f_2(\seq{a}_2)} : \ut}
\infrule
  {\tenv_i \vdash f_i(\seq{a}_i) : \ut \mbox{ (for $i = 1, 2$) }}
  {\tenv_1 + \tenv_2 \vdash \spawnstmt{f_1(\seq{a}_1)}{f_2(\seq{a}_2)} : \ut}
\infrule
  {\tenv + x\COL{}\SenderType + y\COL{}\ReceiverType \vdash f(\seq{a}) : \ut}
  {\tenv \vdash \newstmt{x}{y}{f(\seq{a})} : \ut}
\infrule
  {\tenv_0 \vdash a : \INT \andalso \tenv(x) = \SenderType \andalso \tenv \vdash f(\seq{a}) : \ut}
  {\tenv_0 + \tenv \vdash \sendstmt{a}{x}{f(\seq{a})} : \ut}
\infrule
  {\tenv(y) = \ReceiverType \andalso \tenv, x\COL{}\INT \vdash f(\seq{a}) : \ut}
  {\tenv \vdash \recvstmt{x}{y}{f(\seq{a})} : \ut}
\infrule
  {\tenv(y) = \SenderType \andalso \tenv + x\COL{}\SenderType \vdash f(\seq{a}) : \ut}
  {\tenv \vdash \dupstmt{x}{y}{f(\seq{a})} : \ut}
\infrule
  {\tenv = \{\, f_1 \COL{} (\seq{b}_1) \rightarrow \ut,\ldots,f_n \COL{} (\seq{b}_n) \rightarrow \ut \,\} \\
  \tenv, \seq{x}_i \COL{} \seq{b}_i \vdash M_i : \ut \mbox{ (for each $i \in \{\,1,\ldots,n\,\}$) }}
  {\vdash \SET{f_1(\seq{x}_1\COL{}\seq{b}_1) = M_1,\ldots,f_n(\seq{x}_n\COL{}\seq{b}_n) = M_n} : \tenv}
\infrule
  {\vdash D : \tenv \andalso \tenv \vdash f(\seq{a}) : \ut}
  {\vdash (D, f(\seq{a})) : \ut}
\end{multicols}
\caption{Typing Rules }
\label{fig:target-typing-rules}
\end{figure}

\begin{remark}
  Example~\ref{ex:running-ex} slightly abuses the syntax: some functions drop channels even though their bodies are neither \( () \) nor \( \FAIL \).
  To be more precise, we could spawn a function \( \mathit{drop}(\seq x)=() \) to discard channels;
  however, we omit such explicit drop functions throughout the paper for readability.
  \qed
\end{remark}

\subsection{Operational Semantics}
\label{sec:target-semantics}

The operational semantics is defined by the transition relation
\(C \stepC{D} C' \) on configurations, where a configuration \(C\) is
either \(\FAIL\), representing failure, or a pair \((\Threads, \Qu)\)
consisting of a \emph{thread pool}
and a \emph{channel context}.
A thread pool \(\Threads\) is a finite multiset of function calls of the form
\(f(\seq{a})\). A channel context \(\Qu\in
\Var \rightharpoonup  2^\Var \times \mathbb{Z}^*  \)
maps each receiver channel name (which is represented as a variable)
to a pair consisting of the set of corresponding sender channel names
and a sequence of integers representing the message queue.

\input{sections/os3}

The transition relation \( C \stepC{D} C' \) is defined by the rules
shown in \autoref{fig:target-operational-semantics}.
In the figure, \([\seq{a}/\seq{x}]a'\) denotes the expression
obtained from \(a'\) by substituting \(\seq{a}\) for \(\seq{x}\), and
$a \eval v$ denotes the big-step evaluation relation for simple expressions; its definition is standard and thus omitted.
We briefly explain some important rules.
The rule \rn{R-Send} appends the sent value to the end of the corresponding message queue in the channel context.
This reflects the non-blocking nature of asynchronous communication, where
a sender does not wait for the receiver.
The rule \rn{R-Recv} dequeues a value from the corresponding message queue in the channel context and passes it to the continuation call.
Note that the message queue must be non-empty for \rn{R-Recv} to apply.
The rule \rn{R-Clone} creates a fresh variable \( s \) as an alias for
the sender channel \([\seq{a}/\seq{x}]z\) and relates it to the receiver channel associated with \([\seq{a}/\seq{x}]z\).

%% file: sections/os3.tex
\begin{figure}[t]
  \typicallabel{R-Clone}
  \begin{multicols}{2}
\infrule[R-Fail]
  {f(\seq{x}) = \FAIL \in D}
  {(\Threadsplus{f(\seq{a})}, \Qu) \stepC{D} \FAIL}
\infrule[R-End]
  {f(\seq{x}) = \Vunit \in D}
  {(\Threadsplus{f(\seq{a})}, \Qu) \stepC{D} (\Threads,\Qu)}
  \end{multicols}
  \rulesp
  \infrule[R-Func]
        {f(\seq{x}%
          ) = g(\seq{a}') \in D}
  {(\Threadsplus{f(\seq{a})}, \Qu) \stepC{D} (\Threadsplus{g([\seq{a}/\seq{x}]\seq{a}')}, \Qu)}
\rulesp  
\infrule[R-IfT]
        {f(\seq{x}%
          ) = \ifstmt{a_0}{g_1(\seq{a}_1)}{g_2(\seq{a}_2})\in D \andalso  [\seq{a}/\seq{x}] a_0 \eval n \andalso n \ne 0}
  {(\Threadsplus{f(\seq{a})}, \Qu) \stepC{D} (\Threadsplus{g_1([\seq{a}/\seq{x}]\seq{a}_1)}, \Qu)}
\rulesp  
\infrule[R-IfF]
        {f(\seq{x}%
          ) = \ifstmt{a_0}{g_1(\seq{a}_1)}{g_2(\seq{a}_2})\in D  \andalso [\seq{a}/\seq{x}] a_0 \eval 0} %
  {(\Threadsplus{f(\seq{a})}, \Qu) \stepC{D} (\Threadsplus{g_2([\seq{a}/\seq{x}]\seq{a}_2)}, \Qu)}
\rulesp  
\infrule[R-Send]
        {f(\seq{x}%
          ) = \sendstmt{a_0}{y}{g(\seq{a'})}\in D \andalso \Qu(r) = (S, \qc) \andalso
      [\seq{a}/\seq{x}]y \in S \andalso 
  [\seq{a}/\seq{x}]{a_0} \eval n}
  {(\Threadsplus{f(\seq{a})}, \Qu) \stepC{D} (\Threadsplus{g([\seq{a}/\seq{x}]\seq{a}')}, \Qu[r \mapsto (S, \qc \cdot n)])}
\rulesp  
  \infrule[R-Recv]
          {f(\seq{x}%
            ) = \recvstmt{y}{z}{g(\seq{a}')} \in D \andalso 
            r =[\seq{a}/\seq{x}]z \andalso \Qu(r) = (S, n \cdot \qc)}
  {(\Threadsplus{f(\seq{a})}, \Qu) \stepC{D} (\Threadsplus{g([n/y,\seq{a}/\seq{x}]\seq{a}')}, \Qu[r \mapsto (S, \qc)])}
\rulesp  
\infrule[R-New]
        {f(\seq{x}%
          ) = \newstmt{y}{z}{g(\seq{a}')} \in D \andalso s, r\  \text{fresh}}
  {(\Threadsplus{f(\seq{a})}, \Qu) \stepC{D} (\Threadsplus{g([s/y, r/z,\seq{a}/\seq{x}]\seq{a}')}, \Qu[r \mapsto (\SET{s}, \eseq)])}
\rulesp  
\infrule[R-Clone]
        {f(\seq{x}%
          ) = \dupstmt{y}{z}{g(\seq{a}')} \in D  \quad 
  \Qu(r) = (S, \qc) \quad [\seq a/ \seq x]z \in S \quad s \ \text{fresh}}
{(\Threadsplus{f(\seq{a})}, \Qu)  \stepC{D} (\Threadsplus{g([s/y,\seq a/ \seq x]\seq{a}')}, \Qu[r \mapsto (S \cup \SET{s}, \qc)]}
\rulesp
  \infrule[R-Spawn]
          {f(\seq{x}%
            ) = \spawnstmt{g_1(\seq{a}_1)}{g_2(\seq{a}_2)} \in D \andalso j  \ \text{fresh}}
  {(\Threadsplus{f(\seq{a})}, \Qu) \stepC{D} (\Threadsplus{g_1([\seq{a}/\seq{x}]\seq{a}_1), g_2([\seq{a}/\seq{x}]\seq{a}_2)}, \Qu)}
\caption{Operational Semantics}
\label{fig:target-operational-semantics}
\end{figure}

%% file: sections/Translation3.tex
\section{Reduction to CHC Solving}
\label{sec:Translation}
This section describes our translation from the target language to CHCs.

\subsection{Preliminary: CHCs}
\label{subsec:chc-formalization}
We briefly review constrained Horn clauses (CHCs) with lists and pairs.
The class of CHCs used in this paper is a subclass of CHCs over the combined theory of integer arithmetic and algebraic data types (ADTs).

The sets of \emph{sorts}, \emph{terms}, and \emph{constrained formulas} are defined as follows:
\begin{align*}
  \mbox{ (sorts) } \basesort  ::= &\ \INT \mid \PairTy \INT \INT \mid \ListTy {\INT} \mid \ListTy {(\PairTy \INT \INT)} \\
  \mbox{ (terms) }\term  ::= &\ x \mid n \mid \term_1\ \OP\ \term_2 \mid \NIL \mid \term_1:: \term_2 \mid (\term_1, \term_2) \\
   \mbox { (constraints) } \Constraint ::= &\ \top \mid \bot \mid \term_1 = \term_2 \mid \term_1 < \term_2 \mid \Constraint \lor \Constraint \mid \Constraint \land \Constraint \mid \lnot \Constraint %
\end{align*}
Here, the metavariable \( \OP \) ranges over the same binary arithmetic operations as those used in the definition of the target language.
As usual, \( \NIL \) is the nil list, \( \term_1 :: \term_2 \) is the cons operation, and \( (\term_1, \term_2) \) is the pairing.
The predicate \( < \) is used for integer comparison, and the predicate \( = \) is used for equality between terms of the same sort.

A \emph{constrained Horn clause} (CHC) \( \clause \) is a formula of the form
\begin{gather*}
  \forall \seq{x}:\seq{\basesort}.\ \Head \Leftarrow P_1(\seq{\term}_1) \land \ldots \land P_n(\seq{\term}_n) \land \Constraint
\end{gather*}
where \( \Head \), called the \emph{head}, is either \( \bot \) or \( P_0(\seq{\term}_0) \).
Each \( P_i \) is a \emph{predicate variable} that has an associated type; we write \( P : (\basesort_1, \ldots, \basesort_n ) \) if \( P \) is an \( n \)-ary predicate whose \( i \)-th argument is of sort \( \basesort_i \).
We stipulate that all the CHCs are well-sorted and closed, i.e.~\( \seq{x} \) covers all the free variables appearing in \( \seq{\term}_i \) and \( \Constraint \).
For readability, we may omit the universal quantifiers of a CHC.

A \emph{system of CHCs} is a pair \( (\CHCSet, \sortenv) \) where \( \CHCSet = \{ \clause_1, \ldots, \clause_m \} \) is a finite set of CHCs, regarded as their conjunction, and \( \sortenv = \{P_1, \ldots, P_n\} \) is the set of predicate variables appearing in \( \CHCSet \).
We often write \( \CHCSet \) for \( (\CHCSet, \Delta) \) when \( \Delta \) is clear from the context.
A \emph{solution} of \( (\CHCSet, \{ P_1, \ldots, P_n \}) \) is an interpretation for \( \{P_1, \ldots, P_n \} \) under which \( \CHCSet \) becomes true.
If such a solution exists, we say that \( \CHCSet \) is \emph{satisfiable}.

\begin{example}
  \label{ex:CHC}
  Consider a system of CHCs \((\CHCSet, \set{\forloop\COL(\INT, \INT, \INT, \ListTy{\INT},\ListTy{\INT}), \rloop\COL (\INT, \INT,\INT,\ListTy{\INT})})\), where \(\CHCSet\) consists of the following clauses.
  \begin{align*}
     & \bot \Leftarrow \forloop(\TRUE, 0, n, l, l)\land n\ge 0.\qquad
  \forloop(b, n, n, \NIL, r) \Leftarrow \rloop(b, 0, n, r).\\
  & \forloop(b, i, n, 1::s_1, r) \Leftarrow \forloop(b, i+1, n, s_1, r)\land i<n.\\
  & \rloop(b, c, n,r) \Leftarrow \rloop(b, c+1, n,r')\land r=v::r' \land c+1\le n.\\
  & \rloop(\TRUE, c, n,r) \Leftarrow r=v::r' \land c+1> n.
  \end{align*}
  This is a simplified version of the CHCs in Section~\ref{sec:intro}. It is satisfiable, with the following interpretation as a witness:
  \begin{align*}
    \{&  \forloop \mapsto \set{(\TRUE, i,n,s,r) \mid (0\le i\le n\land |r|>n \land |s|=n-i)\lor n<0},\\
     & \rloop \mapsto \set{(\TRUE, c, n, r) \mid c+|r|>n}\}.
  \end{align*}
  Here, \(|r|\) denotes the length of list \(r\). To see that the interpretation above satisfies the first clause,
  notice that 
  \(\forloop(\TRUE, 0, n, l, l)\) is equivalent to \((0\le 0\le n\land |l|>n \land |l|=n-0)\lor n<0\) under the interpretation,
  which can be simplified to \(n<0\).
\end{example}

\subsection{Translation from Target Language to CHCs}
\label{subsec:translation}
Now we formalize the translation into CHCs.
We first define the translation of types, then describe how programs are translated into systems of CHCs, and finally give the translation rules for individual statements.

\subparagraph{Translation of types}
The encoding \( \transmap{-} \) from argument types to sorts is defined by
\begin{align*}
  \transmap \INT \defeq \INT \qquad \transmap {\ReceiverType} \defeq \ListTy{(\PairTy{\INT}{\INT})} \qquad \transmap{\SenderType} \defeq  \ListTy{(\PairTy{\INT}{\INT})}.
\end{align*}
Channel types are translated into lists of pairs because each channel is represented as a list of timestamped future values to be received or sent on the channel, as explained in the introduction.
Functions are translated into predicates. Accordingly, function types are translated as
\begin{align*}
  \transmap{(\bt_1,\ldots,\bt_n) \rightarrow \ut} \defeq (\INT, \INT, \transmap {\bt_1}, \ldots, \transmap{\bt_n}).
\end{align*}
Note that two additional integer parameters are inserted during the translation.
The first integer argument will be used to track whether a program failure has occurred, and the second will be used to pass timing information, as we shall see below.

\subparagraph*{Translation of programs}
Now we explain how a well-typed program \( (D, f(\seq a)) \) is translated into a system of CHCs.
The core idea is to introduce a predicate \( F_i :  \transmap {\seq{\bt}_i \to \ut} \) for every function \( f_i : {\seq{\bt}_i \to \ut} \) in \( D \), that roughly expresses the following property:
\begin{quote}
  \( F_i(b, t, \seq x_{\INT}, \seq x_{\ct}) \) holds if a call to the function \( f_i \) at time \( t \) with integer values \( \seq{x}_{\INT} \) and channels \( \seq{x}_{\ct} \)\footnote{Here, we assume that integer arguments precede list arguments, although this is not the case in general.}
  satisfies the following conditions.

  \begin{itemize}
    \item it behaves as prophesied by the lists \( \seq x_{\ct} \), that is, the channels send (or receive) the values in the list in the order they appear, following their declared timestamps; and
    \item the flag \( b \) correctly records whether the execution eventually reaches \(\FAIL\).
  \end{itemize}
\end{quote}
Here, \( b = \TRUE \) means that \(\FAIL\) is reachable, whereas \( b = \FALSE \) covers all other cases, including normal termination, divergence, and blocking while waiting for a message.

Based on this idea, we translate each function definition \( f_i(\seq x_i) = M_i\) into a clause \( F_i(b, t, \seq x_i) \Leftarrow \transmap{M_i}^b_t \).\footnotemark
\footnotetext{Strictly speaking,  \( \transmap{M_i}^b_t \) may be a disjunction \( \varphi_{i, 1} \lor \ldots \lor \varphi_{i, m_i} \). In that case, \( F_i(b, t, \seq x_i) \Leftarrow \transmap{M_i}^b_t \) is split into clauses \(  F_i(b, t, \seq x_i) \Leftarrow \varphi_{i,j}\).}
Here we first discuss the two simplest cases, \( M = \FAIL \) and \( M = () \), deferring the full translation rules to later in this subsection.
If \( f_i(\seq x_i) = \FAIL \in D \), then we add the clause
\[
  F_i(b, t, \seq x_i) \Leftarrow b= \TRUE \land \seq s_i = \seq \NIL
\]
where \( \seq s_i \) are variables of type \( \SenderType \) in \( \seq{x}_i \).
(We shall write \( \seq{s}_i = \OutChannels{\seq x_i} \) to denote this condition in the sequel.)
Since the execution of \( f_i \) can reach \(\FAIL\), the failure flag must be \( \TRUE \). Moreover, since no further send operations will be executed, the prophecies for output channels must be empty lists.
In contrast, we impose no such restriction on input channels, because their prophecies represent messages that are or will become available on the channels, rather than messages that are actually received by this execution.
Similarly, if \( f_i(\seq x_i) = () \in D \), then we add the clause
\[
  F_i(b, t , \seq x_i) \Leftarrow b = \FALSE \land \seq s_i = \seq \NIL \qquad \text{where } \seq{s}_i = \OutChannels{\seq x_i}.
\]
The difference is that the failure flag \( b \) is set to \( \FALSE \) because \( () \) means that a thread terminated without any error.

In addition to the clauses \(  F_i(b, t, \seq x_i) \Leftarrow \transmap{M_i}^b_t \),
we add auxiliary clauses for discarding threads that are irrelevant to the safety property.
For each function \( f_i \) in \( D \), we add
\[
  F_i(\FALSE, t, \seq{x}_i) \Leftarrow \seq{s}_i = \seq \NIL \qquad \text{where } \seq{s}_i = \OutChannels{\seq x_i}.
\]
Informally, this clause allows us to ignore the future behavior of a thread from any program point \(f_i\)
at any time \(t\), provided that the thread is assumed not to perform any further sends and not to reach \(\FAIL\).
The reason is that the safety property of interest asks only whether some thread can reach \(\FAIL\);
once such a failure occurs, the whole computation aborts, and the subsequent behavor of the other threads becomes irrelevant.
The condition \(\seq{s}_i = \seq \NIL\) is needed to ensure that the discarded thread contributes no further messages.

Besides the clauses corresponding to function definitions, we add a few more clauses.
We add the clause
\(\bot \Leftarrow F(\TRUE, t, \seq a)\), which states that executing the main function \( f(\seq a) \) never lead to \(\FAIL\).
Finally, we also add clauses for auxiliary predicates, including \( \MERGE \), used in the introduction,
and \(\SORTED \) defined as follows. %
  \begin{align*}
    &\SORTED(\NIL) \Leftarrow \top \qquad \SORTED([(t, v)]) \Leftarrow \top. \\
    &\SORTED((t_1, v_1) :: (t_2, v_2) :: xs) \Leftarrow t_1 < t_2 \land \SORTED((t_2, v_2) :: xs).
  \end{align*}

We write  \( \transmap{(D, f(\seq{a}))} \) for the system of CHCs translated from \( (D, f(\seq a)) \) according to the above procedure.

\begin{figure}[t]
  \begin{align*}
    &\transmap{\sendstmt{a}{x}{f(\seq{a})}}^b_t \defeq \exists t_s \, x' . F(b, t_s, [x' / x]\seq a) \land x = (t_s, a) :: x' \land t \le t_s \\[1ex]
  &\transmap{\recvstmt{x}{y}{f(\seq{a})}}^b_t \defeq \exists t_r \, x \, y' . F(b, t_r, [y'/y]\seq{a}) \land y = (t_s, x) :: y' \land t_s < t_r \land t \le t_r \\[1ex]
    &\transmap{\ifstmt{a}{f_1(\seq{a}_1)}{f_2(\seq{a}_2)}}^b_t \defeq (a \ne 0 \land F_1(b, t, \seq{a}_1)) \lor (a = 0 \land F_2(b, t, \seq{a}_2))\\[1ex]
  &\transmap{\spawnstmt{f_1(\seq{a}_1)}{f_2(\seq{a}_2)}}^b_t \defeq \exists b_1 \, b_2. F_1(b_1, t, \seq a_1) \land F_2(b_2, t, \seq a_2) \land b = (b_1 \lor b_2) \\[1ex]
  &\transmap{\newstmt{x}{y}{f(\seq{a})}}^b_t  \defeq \exists x \, y . F(b, t, \seq a) \land x = y \land \SORTED(x)\\[1ex]
    &\transmap{\dupstmt{x}{y}{f(\seq{a})}}^b_t \defeq  \exists x \, y' . F(b, t, [y'/y]\seq a)
 \land \MERGE(x, y', y)
  \end{align*}
  \caption{Translation of statements. We assume that existentially quantified variables are chosen so as to avoid clashes with free variables in the translated statement.}
\label{fig:translation}
\end{figure}

\subparagraph*{Translation of statements}
We complete the definition of \( F_i(b, t, \seq x_i) \Leftarrow \transmap{M_i}^b_t \) by defining translation \( \transmap{-}^b_t \) from statements (except for \( () \) and \( \FAIL \)) in \autoref{fig:translation}.
Note that %
\( \transmap{-}^b_t \) is parameterized by a failure flag and timestamp.
The timestamp \( t \) is not a global clock:
it records only the most recent send or receive, so as to capture causal dependencies between messages.
We explain several key cases.
In the encoding of a send operation, we check that the prophecy is consistent with the send operation:
the first element of the list must be \( (t_s, a) \), where \( a \) denotes the value being sent and \( t_s \) is the timestamp of that send.
The timestamp \( t_s \) is passed to the predicate \( F \), since it represents the time at which the send operation is performed.
We require \(t \le t_s \), because \(t\) records the most recent send or receive in the current thread.
(Note that only send and receive operations update the timestamp.)
Since the subsequent execution of \( f(\seq a) \) must follow the remaining prophecy, we require \( F(b, t_s, [x'/ x] \seq a) \).
Similarly, the rule for receiving checks that the prophecy for the receiver side is consistent with the receive operation.
The timestamp \( t_r \) represents the time at which the message is received.
We require \(t_r>t_s\), where \(t_s\) is the time at which the message was sent.
This ensures that a message is not received before it is sent, as discussed in~\autoref{ex:dependency}.
In the rule for channel creation, \(x\) and \(y\) represent the sender-side and receiver-side prophecies for future messages, respectively.
The condition \(x=y\) ensures that they coincide, and \(\SORTED(x)\) ensures that the timestamps in \(x\) are in ascending order.
The rule for \texttt{clone} expresses that the prophecy \(y\) is split into two parts: one for the cloned channel \(x\), and
the other for the remaining use of \(y\).
This split is represented by the predicate \( \MERGE(x, y', y) \), which holds if and only if \( y \) is an interleaving of \( x \) and \( y' \); for example \( \MERGE([(1, 2); (5, 4)], [(2, 3)], [(1, 2); (2, 3); (5, 4)]) \) holds.
The predicate \( \MERGE \) is definable using CHCs.

\begin{example}
  \label{ex:running-ex:CHC}
  \newcommand*{\Fn}[2]{F_{#1}(#2)}
  \newcommand*{\For}[1]{\mathit{For}(#1)}
  \newcommand*{\Loop}[1]{\mathit{Loop}(#1)}
  \newcommand*{\Assert}[1]{\mathit{Assert}(#1)}
  \newcommand*{\Fail}[1]{\mathit{Fail}(#1)}
  Recall the program in \autoref{ex:running-ex}.
  We have already seen in the introduction how it is translated into CHCs without a timestamp.
  Below we show the CHCs that are obtained by using the translation rules; we only show the important clauses and omit some clauses such as \( F_i(\FALSE, t, \seq{x}_i) \Leftarrow \seq{c}_i = \seq \NIL \) or the clause for \( \SORTED \).
  \begin{align*}
    &\bot  \Leftarrow \mathit{MsgCount}(\TRUE, t, n) \\
    &\mathit{MsgCount}(b, t, n) \Leftarrow \For {b, t, 0, s r} \land  r =  s \land \SORTED(s) \\
    &\For {b, t, i, n, s, r} \Leftarrow i < n \land \Fn 4 {b, t, i, n, s, r} \\
    &\For {b, t, i, n, s, r} \Leftarrow i \ge n \land \Loop {b, t, 0, n, r} \land s = \NIL \\
    &\Fn 4 {b, t, i, n, s, r} \Leftarrow \MERGE(s_1, s_2, s) \land \Fn 5 {b, t, i, n, s_1, r, s_2}\\
    &\Fn 5 {b_1 \lor b_2, t, i, n, s, r, s'} \Leftarrow \Fn 6 {b_1, t, s'} \land \For {b_2, t, i + 1, n, s, r}\\
    &\Loop {b, t, c, n, r} \Leftarrow r = (t_s, \_v) :: r' \land t_s < t_r \land t \le t_r  \land \Assert {b, t_r, c + 1, n, r'} \\
    &\Assert {b, t, c, n, r} \Leftarrow c \le n \land \Loop{b, t, c, n, r} \\
    &\Assert {b, t, c, n, r} \Leftarrow c > n \land \Fail{b, t, c, n, r} \\ %
    &\Fail{b, t, c, n, r} \Leftarrow b = \TRUE \\
    &\Fn 6 {\FALSE, t, s} \Leftarrow s = [(t_s, 1)] \land t \le t_s \tag*{\qed}
  \end{align*}
\end{example}

\begin{example}
  \label{ex:dependency}
  \definecolor{myRed}{RGB}{255,0,51}
  \newcommand*{\tpart}[1]{\textcolor{myRed}{#1}}
  \newcommand{\fnsymb}[1]{f_{#1}}
  \newcommand*{\fn}[2]{\fnsymb{#1}(#2)}
  \newcommand*{\Fn}[2]{F_{#1}(#2)}
  To illustrate the need for timestamps we consider the following simple program (on the left-hand side) and the corresponding CHCs (on the right-hand side).
  {\setlength{\mathindent}{0pt}
  \begin{minipage}{.4\linewidth}
    \begin{align*}
    &\fn 0 {} = \newstmt s r {\fn 1 {s, r}} \\
    &\fn 1 {s,r} = \recvstmt x r {\fn 2 s} \\
    &\fn 2 {s} = \sendstmt 0 s {\fn 3 {}} \\
    &\fn 3 {} = \FAIL \\
  \end{align*}
\end{minipage}
\begin{minipage}{.5\linewidth}
  \begin{align*}
    &\Fn 0 {b, \tpart t} \Leftarrow \Fn 1 {b, \tpart t, s, r} \land r = s \tpart{{} \land \SORTED(s)} \\
    &\Fn 1 {b, \tpart t, s, r} \Leftarrow r = [(\tpart {t_s},x)]  \land \Fn 2 {b, \tpart {t_r}, s} \tpart{{} \land t_s < t_r \land t \le t_r} \\
    &\Fn 2 {b, \tpart t, s} \Leftarrow s = [(\tpart {t_s},0)] \land \Fn 3 {b, \tpart{t_s}} \tpart{{} \land t \le t_s} \\
    &\Fn 3 {b, \tpart t} \Leftarrow b = \TRUE \\
    &\bot \Leftarrow \Fn 0 {\TRUE, \tpart t}
  \end{align*}
  \end{minipage}
} %
\par

The program on the left-hand side never fails because the execution is blocked by the receive from \( r \).
However, if we ignore the timestamps, that is, omit the red parts of the CHCs, then the resulting system of CHCs on the right-hand side becomes unsatisfiable, which means that such a translation is incomplete.
Indeed, by repeatedly applying the clauses backwards, we obtain
  \[
    \bot \Leftarrow \Fn 0 {\TRUE} \Leftarrow \Fn 1{\TRUE, [0], [0]} \Leftarrow \Fn 2 {\TRUE, [0]} \Leftarrow \Fn 3 {\TRUE} \Leftarrow \top.
    \]
    Thus, \(\bot\) is derivable.
    The prophecy \( [0] \) above records that the message \(0\) would be sent on \(s\),
    but does not capture the causal dependency between the receive and send operations.
  On the other hand, with timestamps, a similar derivation would contain the step
  \[
    \Fn 1 {\TRUE, \tpart {t_0}, [(\tpart {t_s}, 0)], [(\tpart {t_s}, 0)]} \Leftarrow \Fn 2 {\TRUE, \tpart {t_r}, [(\tpart {t_s}, 0)]} \text{ with constraints } \tpart{t_s < t_r \land t_r < t_s \land t_0 \le t_r}.
  \]
  Here, the constraint \( \tpart{t_r < t_s} \) expresses the sequential causality between
  \( \fnsymb  1 \) and \( \fnsymb 2\),
  whereas \( \tpart{t_s < t_r} \) comes from the causal dependency induced by the channel.
  Since these constraints are contradictory, we cannot derive \( \bot \).
  \qed
\end{example}

\subsection{Properties of the Translation}
\label{subsec:translation-properties}

Our translation is sound and complete in the sense that a program is free from failure if and only if
the corresponding system of CHCs is satisfiable.
We write \( C \neverStepC{D} C' \) if there is no reduction sequence from \( C \) to \( C' \).
\begin{theorem}[Soundness and Completeness]
\label{thm:soundness}
Let $\vdash (D, f(\seq{a})) : \ut$.
Then \(  \transmap{(D, f(\seq a))}\) is a well-sorted system of CHCs that is satisfiable if and only if $(\SET{f(\seq{a})}, \emptyset) \neverStepC{D} \FAIL$.
\qed
\end{theorem}
The proof is given in \autoref{sec:proof-of-properties}.
Here we provide a brief overview of the argument.
The core argument is a correspondence between (i) an execution leading to \(\FAIL\) and
(ii) a derivation of contradiction from the translated CHCs.
A subtlety is that such a derivation generally has a tree structure rather than a sequence,
because the clause for \texttt{spawn} is non-linear.
Recovering an execution therefore requires linearizing this tree into
a thread schedule while preserving causality. The timestamp
information is precisely what makes this possible.

%% file: sections/ImplementationAndEvaluation.tex
\section{Implementation and Experiments}
\label{sec:implementation-evaluation}
We implemented a verification tool based on our reduction and conducted preliminary experiments on small benchmarks to examine the effectiveness of our approach.
Our tool takes as input a program written in the language defined in \autoref{sec:target-language}.
Then it generates CHCs, %
either with or without timestamps, by applying our reduction described in \autoref{sec:Translation}.
The tool first passes the CHCs with timestamps to a portfolio CHC solver that runs \spacer{}~\cite{DBLP:journals/fmsd/KomuravelliGC16}, \eldarica{}~\cite{hojjat2018eldarica}, and \catalia{}~\cite{katsura2025automated} in parallel.\footnotemark
\footnotetext{The versions of the solvers used in our experiments are as follows: \spacer{} 4.15.5, \eldarica{} 2.2.1, and \catalia{} with commit hash 227e6b4ae6870cd795a7dbc4a863c45dcb900945.}
If one of the solvers returns \texttt{sat} (resp.~\texttt{unsat}), the tool reports \texttt{safe} (resp. \texttt{unsafe}).
If all solvers time out, the tool instead passes the CHCs without timestamps, as an approximate translation, to the same portfolio solver.
If any solver reports that the approximate CHCs are satisfiable, the tool returns \texttt{safe}; otherwise, it returns \texttt{unknown}.

We now describe the experiments.
We ran the tool on our benchmark suite, which consists of \autoref{ex:running-ex}, \autoref{ex:dependency}, and the programs in \autoref{appx:evaluation-examples}. \autoref{ex:dependency} and \texttt{ack} are the examples for which the results differ depending on whether timestamps are included. Instances whose names have the suffix ``-e'' are obtained by modifying the corresponding instances without the suffix so that an assertion fails.
All the benchmark instances were run on a commodity laptop (2.8GHz Intel Core i7 with 16GB RAM).
The timeout for the portfolio solver was set to 120 seconds.

\begin{table}[h]
  \centering
  \begin{minipage}[t]{0.48\linewidth}
    \centering
    \captionof{table}{Results for safe cases}
    \label{tab:evaluation-results}
    \begin{tabular}{lc}
      \hline
      \textit{Instance} & \textit{Solving Time} \\
      \hline
      \autoref{ex:running-ex} & 19.5 \\
      \autoref{ex:dependency} & $< 0.1$ \\
      ack & 0.2 \\
      multi-sends & 50.4($+120$) \\
      client-server & \timeout \\
      calc-server & 0.8 \\
      \hline
    \end{tabular}
  \end{minipage}
  \hfill
  \begin{minipage}[t]{0.48\linewidth}
    \centering
    \captionof{table}{Results for unsafe cases}
    \label{tab:evaluation-results-fail}
    \begin{tabular}{lc}
      \hline
      \textit{Instance} & \textit{Solving Time} \\
      \hline
      \autoref{ex:running-ex}-e & $< 0.1$ \\
      ack-e & 0.1 \\
      multi-sends-e & $< 0.1$ \\
      client-server-e & $< 0.1$ \\
      calc-server-e & 0.1 \\
      \hline
    \end{tabular}
  \end{minipage}
\end{table}

\autoref{tab:evaluation-results} and \autoref{tab:evaluation-results-fail} show the experimental results. \autoref{tab:evaluation-results} lists the safe-case benchmarks, and \autoref{tab:evaluation-results-fail} lists the unsafe-case benchmarks.
In both tables, the ``Solving time'' column shows the time taken by the portfolio solver to return a result.
Here, \texttt{timeout} means that the portfolio solver timed out on both the CHCs with timestamps and those without timestamps after 120 seconds, and ``($+120$)'' indicates that it timed out on the CHCs with timestamps.
\autoref{tab:evaluation-results} shows that, except for \texttt{client-server}, all benchmarks in the safe set were verified
to be safe within the time limit.  \autoref{tab:evaluation-results-fail} shows that
unsafety was detected in all buggy benchmark programs.

The \texttt{client-server} instance illustrates a limitation of the current backend CHC solvers
rather than a fundamental limitation of our reduction.
The program \texttt{client-server} is shown in \autoref{fig:client-server}. In the function \texttt{main}, a spawned thread acts as a server that adds $1$ to the received value and sends back the result. Another thread repeatedly sends $1$, receives the response,
and repeats this $n$ times, asserting that the sum of the received values is equal to $2n$.
The reason this instance is not verified is that the backend CHC solvers fail to discover the invariants stating that the values received
through \texttt{r1} and \texttt{r2} are always \( 1 \) and \( 2 \), respectively.
More precisely, \eldarica{} and \spacer{} cannot synthesize predicates involving recursively defined predicates such as ``the list only contains 1''.
On the other hand, in principle, \catalia{} can synthesize such predicates as long as they are expressible as arithmetic formulas over integer values obtained by applying a certain class of catamorphisms to lists.
In practice, however, \catalia{} fails to synthesize a suitable predicate because it relies on heuristics to infer the catamorphism to use.
Thus, this example suggests the need for further improvement of backend CHC solvers, especially in their support for synthesizing recursive predicates.

\begin{figure}[tbp]
\begin{minipage}[t]{.5\linewidth}
\begin{verbatim}
fn main(n: usize) {
    let (s1, r1) = channel();
    let (s2, r2) = channel();
    spawn(move || {
        loop {// #[invariant(r2.recv = 1)]
            let v = r2.recv().unwrap();
            s1.send(v + 1).unwrap();
        }
    });
\end{verbatim}
\end{minipage}
\qquad
\begin{minipage}[t]{.4\linewidth}
\begin{verbatim}
let mut sum = 0;
    for i in 0..n {
        s2.send(1).unwrap();
        let v = r1.recv().unwrap();
        sum += v;
    }
    assert!(sum == n * 2);
} // end of main
\end{verbatim}
\end{minipage}
\caption{Client-server program. (A single program is shown over two columns.)}
\label{fig:client-server}
\end{figure}

A possible way to address this limitation is to make use of user-provided annotations.
We have confirmed that if we manually provide an appropriate catamorphism---namely, one that counts the numbers of \( 1 \)s and \( 2 \)s appearing in a list---then \catalia{} can verify that \texttt{client-server} is safe.\footnote{The catamorphism we gave was \texttt{foldr f (0, 0)} where \texttt{f = \textbackslash x (n1, n2) -> case x of \{1 -> (n1 + 1, n2); 2 ->  (n1, n2 + 1); \_ -> (n1, n2)\}}. }
Directly providing a catamorphism as a hint would impose a burden on users who are not familiar with the backend CHC solvers.
A more practical approach would be to extend our tool so that annotations in the source program, such as the comment in \autoref{fig:client-server}, are translated into hints that help \catalia{} synthesize the required catamorphism.

%% file: sections/Discussion.tex
\section{Discussion}
\subsection{Variations of Message-Passing Semantics}
\label{sec:discussion}
Our prophecy-based approach is not tied to the particular message-passing semantics considered so far: it can also be adapted to other message-passing semantics with only modest changes to the translation.
The semantics of the target language in \autoref{fig:target-operational-semantics} is based on Rust's standard mpsc library.
In this semantics, messages are enqueued in the FIFO buffer owned by the corresponding receiver in the order in which they are sent, and hence the send order is preserved.
In settings such as distributed computing, however, it is also natural to consider the following two variants:
\begin{enumerate}
  \item a semantics in which the order of messages sent by each individual sender is preserved but no ordering is imposed between messages sent by different senders associated with the same receiver; and
  \item a semantics in which the order of messages sent is not preserved at all, as in the unordered buffering modeled by the asynchronous \(\pi\)-calculus~\cite{DBLP:conf/birthday/BeauxisPV08}.
\end{enumerate}

The two variants can be captured by modifying how message buffers are represented.
Preserving the order of messages only within each sender channel suggests maintaining a separate FIFO queue for each sender channel.
In our setting, however, sender channels can be cloned, and this cloning structure must also be reflected in the buffer.
This leads to a tree-structured buffer: each leaf corresponds to a sender channel, and when a sender is cloned, the corresponding leaf is replaced by an internal
node with two children for the original sender and the newly created clone.
The internal node stores the queue of messages sent before the cloning, while the leaves store the messages sent after the cloning.

To model unordered buffering, we instead use multisets: send operations add elements to the multiset, and receive operations remove elements from it nondeterministically.

We illustrate the differences among these three semantics using the example in \autoref{fig:semantics-difference-example}.
With the original semantics in \autoref{fig:target-operational-semantics}, after all send operations, the buffer state is the FIFO queue $[0, 1, 2]$.
Therefore, the values are received in the order $0, 1, 2$.
If we instead use the tree-structured buffer, the buffer state is $\mathsf{Node}([0], s_1 \mapsto [1], s_2 \mapsto [2])$, where the root queue $[0]$ stores the messages sent before cloning, and each leaf $s_1 \mapsto [1]$ and $s_2 \mapsto [2]$ represents the sender and the messages sent via the sender after cloning.
Thus, the value $0$ is received first, after which either $1$ or $2$ may be received next.
Hence the possible receive orders are $0, 1, 2$ and $0, 2, 1$.
Finally, with multiset buffers, the buffer state after the sends is the multiset \( \{0,1,2 \} \).
Since a multiset does not preserve order, the values may be received in any order.

\begin{figure}[tbp]
\begin{verbatim}
fn main() {
    let (s1, r) = channel();
    s1.send(0).unwrap();
    let s2 = s1.clone();
    s1.send(1).unwrap();
    s2.send(2).unwrap();
    let v1 = r.recv().unwrap();
    let v2 = r.recv().unwrap();
}
\end{verbatim}
\caption{Example program illustrating the differences among the three semantics.}
\label{fig:semantics-difference-example}
\end{figure}

We next describe how the translation changes for these two variants.
The tree structure is already reflected by the use of \( \MERGE \) at \texttt{clone} operations: each clone splits a sender prophecy into two branches.
Thus, the only change needed is to omit the constraint \( \SORTED (x) \), which otherwise imposes a global FIFO order on all messages sent to the receiver.
For the multiset semantics, we need to represent prophecies using multisets rather than lists. Accordingly, the auxiliary predicates must be modified. The predicate $\MERGE(x, y', y)$ is replaced by $\Union(x, y', y)$, expressing that the union of $x$ and $y'$ is $y$. Since $\SORTED(x)$ has no meaning for multisets, it is omitted.
The soundness and completeness proofs require only local changes. %

As another variant of the semantics, we can also handle bounded buffers by adding receive-time information to prophecies.
Assume that the buffer size is bounded by $k$. 
In this case, the sender and receiver prophecies are extended to contain receive-time timestamps in addition to send-time timestamps, and the $\SORTED$ predicate additionally imposes the constraint that, for each $i$, the send-time timestamp of the $(i+k)$-th message is later than the receive-time timestamp of the $i$-th message. 
This constraint states that no further value is sent when there are already $k$ unreceived values accumulated in the buffer, which is exactly the buffer-size constraint.

\subsection{Deadlock Freedom}

Deadlock freedom is a classical property of concurrent programs. Our approach can handle deadlock freedom by adjusting our translation to CHCs.

We illustrate the translation to ensure deadlock freedom through the example in \autoref{fig:deadlock-freedom-example}. The example program in Figure 7 is a typical example of a deadlock: the two threads each wait to receive a value from the other.

\begin{figure}[tbp]
\newcommand{\fnsymb}[1]{f_{#1}}
\newcommand*{\fn}[2]{\fnsymb{#1}(#2)}
\newcommand*{\Fn}[2]{F_{#1}(#2)}
\newcommand*{\gn}[2]{g_{#1}(#2)}
\newcommand*{\Gn}[2]{G_{#1}(#2)}
\begin{minipage}{.45\linewidth}
  {\small\begin{align*}
    &\fn 0 {} = \newstmt {s_1} {r_1} {\fn 1 {s_1, r_1}} \\
    &\fn 1 {s_1, r_1} = \newstmt {s_2} {r_2} {\fn 2 {\seq s, \seq r}} \\
    &\fn 2 {\seq s, \seq r} = \spawnstmt {\gn 0 {s_1, r_2}} {\fn 3 {r_1, s_2}} \\
    &\gn 0 {s_1, r_2} = \recvstmt {v} {r_2} {\gn 1 {s_1}} \\
    &\gn 1 {s_1} = \sendstmt {0} {s_1} {()} \\
    &\fn 3 {r_1, s_2} = \recvstmt {v} {r_1} {\fn 4 {s_2}} \\
    &\fn 4 {s_2} = \sendstmt {0} {s_2} {()}
  \end{align*}}
\end{minipage}
\begin{minipage}{.45\linewidth}
{
\small
\begin{align*}
  &\bot \Leftarrow \Fn 0 {B, t}\\
  &\Fn 0 {b, t} \Leftarrow \Fn 1 {b, t, s_1, s_1} \land \SORTED(s_1)\\
  &\Fn 1 {b, t, s_1, r_1} \Leftarrow \Fn 2 {b, t, \seq s, r_1, s_2} \land \SORTED(s_2) \\
  &\Fn 2 {b, t, r, \seq s, \seq r} \Leftarrow 
  \begin{aligned}[t]
    &\Gn 0 {b_1, t, s_1, r_2} \land \Fn 3 {b_2, t, r_1, s_2} \\
    &\land b = b_1 \combine b_2
  \end{aligned} \\
  &\Gn 0 {b, t, s_1, r_2} \Leftarrow
  \begin{aligned}[t]
    &\Gn 1 {b, t_r, s_1} \land r_2 = (t_s, v) :: r'_2 \\
    &\land t_s < t_r \land t \le t_r
  \end{aligned} \\
  &\Gn 0 {B, t, \NIL, r_2} \Leftarrow r_2 = \NIL\\
  &\Gn 1 {T, t, s_1} \Leftarrow s_1 = [(t_s, 0)] \land t \le t_s\\
  &\Fn 3 {b, t, r_1, s_2} \Leftarrow
  \begin{aligned}[t]
    &\Fn 4 {b, t_r, s_2} \land r_1 = (t_s, v) :: r'_1 \\
    &\land t_s < t_r \land t \le t_r
  \end{aligned} \\
  &\Fn 3 {B, t, r_1, \NIL} \Leftarrow r_1 = \NIL\\
  &\Fn 4 {T, t, s_2} \Leftarrow s_2 = [(t_s, 0)] \land t \le t_s
\end{align*}
}
\end{minipage}
\caption{Example program that causes a deadlock and its translation.}
\label{fig:deadlock-freedom-example}
\end{figure}

In the translation described in \autoref{sec:Translation}, the flag \( b \) is a Boolean flag that indicates whether the program will fail. 
To check deadlock freedom, we extend this flag to three-valued status: $\FailStatus$, $\Terminated$, and $\Blocked$. The statuses $\FailStatus$ and $\Terminated$ represent reaching \(\FAIL\) and \(()\), respectively.
The status $\Blocked$ indicates that a thread is blocked at a subsequent $\RECV$ statement because no value is sent by the corresponding sender. 
In the CHCs, each status is denoted by its initial letter.

We make the following changes to the translation from programs to CHCs. First, the translation rule for \(\RECV\) statements is extended with a case in which an empty receiver prophecy causes the flag to be set to \(\Blocked\).
In the translation shown in \autoref{fig:deadlock-freedom-example}, the following CHCs correspond to this extension.
\begin{align*}
  &G_0(B, t, \NIL, r_2) \Leftarrow r_2 = \NIL \quad F_3(B, t, r_1, \NIL) \Leftarrow r_1 = \NIL
\end{align*}
Next, at a $\SPAWN$ statement, the statuses of two threads are combined so that the overall computation is reported as \(\Blocked\) when both remaining threads are either terminated or blocked and at least one of them is blocked.
In the CHCs, the operator \(\combine\) represents this composition and is defined as follows.
\begin{align*}
  T \combine T = T \quad  D \combine D = D \combine T = T \combine D = D \quad F \combine \_ = \_ \combine F = F
\end{align*}
We modify the goal clause so that it derives \(\FALSE\) whenever the status of the whole program is \(\Blocked\).
Consequently, the system of CHCs is unsatisfiable if the program can deadlock.
In the example in Figure 7, both threads have the status \(\Blocked\) at the \(\SPAWN\) point. The status of the entire program is thus \(\Blocked\), signaling that a deadlock can occur.

%% file: sections/Related.tex
\section{Related Work}
\label{sec:rel}

Type-based approaches have been popular for the automated verification
and analysis of message-passing concurrent programs: see
\cite{kobayashi2003type,huttel2016foundations} for surveys.  While
such approaches are often lightweight, they are generally not complete
and may yield false alarms.  Toby and
Ankush~\cite{10.1007/978-3-032-22723-2_11} recently proposed a type
inference algorithm for refinement binary session types, which can
automatically verify the absence of assertion failures.  However,
binary session types are designed for one-to-one communication and
do not directly support multiple senders.
Moreover, their inference
method is template-based and does not enjoy relative completeness. By
contrast, our translation is relatively complete, assuming that the
underlying CHC solver is complete; this assumption is, of course,
idealized, since CHC solving is undecidable in general.

Model checking~\cite{ClarkeModelChecking,HandbookMC,PrinciplesMC} has
also been widely used for the automated verification of concurrent
systems. However, standard finite-state or pushdown model checking is
not directly applicable to systems with 
unbounded numbers of threads and unbounded message queues. There are
decidable classes of message-passing concurrent systems, such as lossy
channel systems~\cite{DBLP:conf/lics/AbdullaJ93}, but our target
language is Turing-complete and therefore cannot be captured by such
models. Bounded model checking is often used for infinite state systems,
but it is unsound in general.

Our use of prophecies~\cite{abadi1991existence} to enable modular modeling of sender and
receiver processes is related in spirit to thread-modular verification for
shared-memory multi-threaded
programs~\cite{DBLP:conf/esop/FlanaganFQ02,DBLP:conf/cav/HenzingerJMQ03,10.1007/978-3-540-74061-2_14}. In
thread-modular verification, suitable assumptions about inter-thread
interactions are inferred and then used to verify each thread
modularly. Although shared-memory and message-passing concurrency
differ significantly, our prophecies about future messages play a role
analogous to such assumptions.
Of particular note is the introduction of prophecies into the Iris separation logic~\cite{jung2015iris} by Jung et al.~\cite{jung2019future}, used for modular verification of logical atomicity (a variant of linearizability) of concurrent data structures.
Although their mechanism supports various types of concurrency, unlike our work, theirs is not designed for automation and does not support partial resolution of prophecies, which our work uses to partially determine the contents of the prophesied list.
We also note that prophecies were first introduced to prove refinement between state machines~\cite{abadi1991existence}.
Our use of prophecies for reducing a stateful program to a system of CHCs might be relevant to this in spirit.

Our verification technique was inspired by the prophecy-based
technique of RustHorn by Matsushita et al.~\cite{matsushita2021rusthorn}, which has also been applied to semi-automated Rust verification tools such as Creusot~\cite{denis2022creusot}, Verus~\cite{lattuada2023verus} and Thrust~\cite{ogawa2025thrust},
and has been justified semantically in the Iris separation logic~\cite{matsushita2022rusthornbelt,matsushita2025nola,travis2026verusbelt}.
They considered
\emph{sequential} programs with mutable references, representing a
borrowed mutable reference as a pair of its current value and the
future value at the point where the borrow is released. 
Our technique was inspired by theirs.
Both our work and RustHorn both prophesy future interactions between entities (channel senders and receivers, or borrowers and lenders) for automated
modular functional reasoning.
However, our work is novel in that it supports multi-shot, multi-sender channels, where a list of multiple sent values is prophesied, and the send/receive events can interleave in a complex way depending on thread scheduling.
Roughly speaking, Rust-style borrows are like a one-shot, single-sender channel.

Automated verification based on CHCs~\cite{Bjorner15} or, more generally, fixpoint
logic~\cite{DBLP:conf/esop/0001TW18,DBLP:journals/pacmpl/BurnOR18} has attracted considerable attention in recent years. Various
verification problems can be encoded and solved using common backend
solvers for CHCs~\cite{DBLP:journals/fmsd/KomuravelliGC16,hojjat2018eldarica,katsura2025automated,DBLP:conf/aplas/Champion0S18}
or fixpoint logics~\cite{DBLP:conf/sas/0001NIU19,DBLP:journals/pacmpl/KobayashiTST23,DBLP:journals/pacmpl/UnnoTGK23}.
Most such verification techniques have primarily been developed for sequential programs, and, to the
best of our knowledge, the CHC-based verification of message-passing
concurrent programs has not been studied before.

%% file: sections/Conc.tex
\section{Conclusion}
\label{sec:conc}

We proposed a novel method for verifying the functional correctness of
message-passing concurrent programs. Inspired by RustHorn's
prophecy-based technique, we represented each communication channel as
a list of future messages to be sent on that channel, thereby reducing
the verification problem to CHC solving while preserving soundness and
completeness. We also implemented a prototype tool to demonstrate the
usefulness of our method.

%% file: sections/TranslationAppx.tex
\section{Supplementary Materials for Section~\ref{sec:Translation}}
\label{appx:Translation}

\subsection{Definition of \texorpdfstring{\( \SORTED \)}{Sorted} and \texorpdfstring{\( \MERGE \)}{Merge}}
The clauses for the predicates \( \SORTED \) and \( \MERGE \), which are used in the translation, are given as follows:
\begin{align*}
  &\SORTED(\NIL) \Leftarrow \top \\
  &\SORTED([(t, v)]) \Leftarrow \top \\
  &\SORTED((t_1, v_1) :: (t_2, v_2) :: xs) \Leftarrow t_1 < t_2 \land \SORTED((t_2, v_2) :: xs) \\
  \\
  &\MERGE(\NIL, \NIL, \NIL) \Leftarrow \TRUE\\
  &\MERGE(p :: xs, ys, p :: zs) \Leftarrow \MERGE(xs, ys, zs)\\
  &\MERGE(xs, p :: ys, p :: zs) \Leftarrow \MERGE(xs, ys, zs)\
\end{align*}

\subsection{Examples}
The following example shows a program with causal dependencies between channels that are used across different threads.
\begin{example}
  \label{ex:dependency-threads}
  \newcommand{\fnsymb}[1]{f_{#1}}
  \newcommand*{\fn}[2]{\fnsymb{#1}(#2)}
  \newcommand*{\Fn}[2]{F_{#1}(#2)}
  \newcommand*{\gn}[1]{g(#1)}
  \newcommand*{\Gn}[1]{G(#1)}
  \newcommand*{\hn}[2]{h_{#1}(#2)}
  \newcommand*{\Hn}[2]{H_{#1}(#2)}

  \newcommand*{\ack}{\mathit{ack}}
  The following program\footnote{
    The example relaxes the syntax and typing rules of the language for readability; for example, we drop variables before reaching \( () \).}
  sends messages across three threads \( f \), \( g \) and \( h \) as illustrated in the left-hand side of \autoref{fig:message-seq} and checks whether the values received by \( r \) are \( 10 \) and then \( 20 \).
  The corresponding CHCs are given side-by-side.

  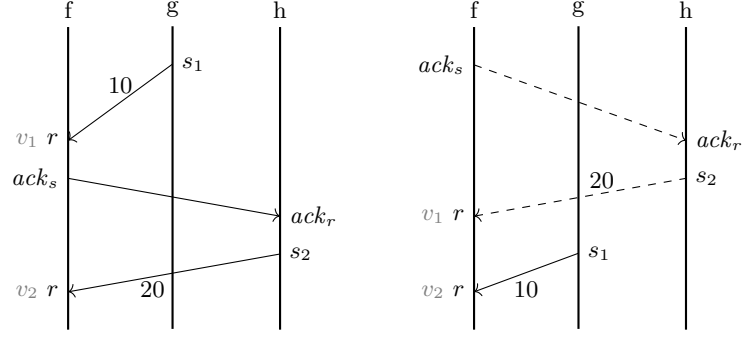
\begin{figure}
  \centering
  \input{figures/message_seq}
  \caption{Message sequence diagrams for~\autoref{ex:dependency-threads}: the left shows a correct sequence, the right an incorrect one where the dotted lines are communications that cannot happen.}
  \label{fig:message-seq}
\end{figure}

  \par\noindent
  {\setlength{\mathindent}{0pt}
    \begin{minipage}{.45\linewidth}
  {\small
    \begin{align*}
    \\
    &\fn 0 {} = \newstmt {s_1} r {\fn 1 {s_1, r}} \\
    &\phantom{\SORTED}\\
    &\fn 1 {s_1, r} = \newstmt {\ack_s}{\ack_r}  {\fn 2 {s_1, r, \seq \ack}} \\
    &\phantom{\SORTED}\\
    &\fn 2 {s_1, r, \seq \ack} = \dupstmt {s_2} {s_1} {\fn 3 {\seq s, r, \seq \ack}} \\
    &\fn 3 {\seq s, r, \seq \ack} = \spawnstmt {\gn {s_1}}  {\fn 4 {s_2, r, \seq \ack}} \\
    &\gn {s_1} = \sendstmt {10} {s_1} {()}\\
    &\fn 4 {s_2, r, \seq \ack} = \spawnstmt {\hn 1 {s_2, \ack_r}}  {\fn 5 {r, \ack_s}} \\
    &\hn 1 {s_2, \ack_r} = \recvstmt {\_} {\ack_r} {\hn 2 {s_2}}\\
    \\
    &\hn 2 {s_2} = \sendstmt {20} {s_2} {()}\\
    &\fn 5 {r, \ack_s} = \recvstmt {v_1} r {\fn 6 {r, \ack_s, v_1}} \\
    \\
    &\fn 6 {r,\ack_s, v_1} = \sendstmt 0 {\ack_s} {\fn 7 {r, v_1}} \\
      \\
    &\fn 7 {r, v_1} =  \recvstmt {v_2} r {\fn 8 {v_1, v_2}}\\
    \\
    &\fn 8 {v_1, v_2} = \mathtt{assert}(v_1= 10 \mathrel{\&\&} v_2 = 20)
  \end{align*}
} %
\end{minipage}
\begin{minipage}{.45\linewidth}
{\small
  \begin{align*}
    &\bot \Leftarrow  \Fn 0 {\TRUE, t}\\
    &\Fn 0 {b, t} \Leftarrow
      \begin{aligned}[t]
      &\Fn 1 {b, t, s, r} \land r = s \\
      &\land \SORTED(s)
    \end{aligned}\\
    &\Fn 1 {b, t, s, r} \Leftarrow
      \begin{aligned}[t]
      &\Fn 2 {b, t,s, r, \seq \ack} \land \ack_r = \ack_s \\
      &\land \SORTED(\ack_s)
      \end{aligned}\\
    &\Fn 2 {b, t, s, r, \seq \ack} \Leftarrow \MERGE(s_1, s_2, s) \land \Fn 3 {b, t, \seq s, r, \seq \ack}\\
    &\Fn 3 {b_1 \lor b_2,  t, \seq s, r, \seq \ack} \Leftarrow \Gn {b_1, t, s_1} \land \Fn 4 {b_2, t, s_2, r, \seq \ack} \\
    &\Gn {\FALSE, t, s_1} \Leftarrow s_1 = [(t_{s_1}, 10)] \land t \le t_s \\
    &\Fn 4 {b_1 \lor b_2, t, s_2, r, \seq \ack} \Leftarrow \Hn 1 {b_1, t, s_2, \ack_r} \land  \Fn 5 {b_2, t, r, \ack_s} \\
    &\Hn 1 {b, t, s_2, \ack_r} \Leftarrow
      \begin{aligned}[t]
        &\ack_r = [(t_{\ack_s}, \scalebox{0.5}[1]{\text{\_}})]  \land \Hn 2 {b, \textcolor{red}{t_{\ack_r}}, s_2, \ack_r'} \\
        &\land t_{\ack_s} < t_{\ack_r} \land t \le t_{\ack_r}
      \end{aligned}\\
    &\Hn 2 {\FALSE, \textcolor{red}{t}, s_2} \Leftarrow s_2 = [(\textcolor{red}{t_{s_2}}, 20)] \land \textcolor{red}{t \le t_{s_2}} \\
    &\Fn 5 {b, t, r, \ack_s} \Leftarrow
      \begin{aligned}[t]
        &r = (t_{s_1}, v_1) :: r' \land \Fn 6 {b, t_r, r', \ack_s, v_1} \\
        &\land t_{s_1} < t_{r} \land t \le t_r
      \end{aligned}\\
    &\Fn 6 {b, t, r,\ack_s, v_1} \Leftarrow
      \begin{aligned}[t]
        &\ack_s = [(t_{\ack_s}, 0)] \land {\Fn 7 {r, v_1}} \\
        &\land t \le t_{\ack_s}
      \end{aligned} \\
    &\Fn 7 {b, t, r, v_1} \Leftarrow
      \begin{aligned}[t]
        &r = [(t_{s_2}, v_2)] \land \Fn 8 {b, t_r, v_1, v_2} \\
        &\land  t_{s_2} < t_r \land t \le t_r
      \end{aligned}\\
    &\Fn 8 {b, t, v_1, v_2} \Leftarrow b \neq (v_1= 10 \land v_2 = 20)
  \end{align*}
   } %
 \end{minipage}
 }%

 The assertion never fails because in order for \( h \) to send the value \( 20 \), (1) the value \( 10 \) must be sent before a message is sent using \( \ack_s \); (2) \( \ack_r \) must wait for \( \ack_s \); and (3) a message must be received via \( \ack_r \) for \( h \) to send \( 20 \).
 The red part of the CHCs imposes an order \( t_{\ack_r} \le t_{s_2}\) which captures the last causal dependency; the other two dependencies \(t_{s_1} < t_{\ack_s} \) and \( t_{\ack_s} < t_{\ack_r}\) can similarly be read from the CHCs.

 If there were no timestamps, then \( s \mapsto [20, 10] \) would become a valid prophecy.
 This prophecy does not capture the correct behavior as it corresponds to a situation where the send using \( s_2 \) happened before the send using \( s_1 \) as, for example, in the right-hand side of \autoref{fig:message-seq}.
 \qed
\end{example}

%% file: figures/message_seq.tex
\begin{tikzpicture}[every node/.style={font=\small}, thread/.style={thick}]

\node (main) {f};
\node (t1) [right=of main] {g};
\node (t2) [right=of t1] {h};

\draw[thread] (main.south) -- ++(0,-4);
\draw[thread] (t1.south) -- ++(0,-4);
\draw[thread] (t2.south) -- ++(0,-4);

\coordinate (m1) at ($(main.south)+(0,-1.5)$);
\coordinate (m2) at ($(main.south)+(0,-2)$);
\coordinate (m3) at ($(main.south)+(0,-3.5)$);

\coordinate (t11) at ($(t1.south)+(0,-0.5)$);
\coordinate (t21) at ($(t2.south)+(0,-2.5)$);
\coordinate (t22) at ($(t2.south)+(0,-3)$);

\draw[<-] (m1) -- node[above] {10} (t11);
\node[left] at (m1) {\textcolor{gray}{\( v_1\)} \( r \)};
\node[right] at (t11) {\( s_1 \)};

\draw[->] (m2) -- (t21);
\node[left] at (m2) {\( \mathit{ack}_s \)};
\node[right] at (t21) {\( \mathit{ack}_r \)};

\draw[<-] (m3) -- node[below left] {20} (t22);
\node[left] at (m3) {\textcolor{gray}{\( v_2\)} \( r \)};
\node[right] at (t22) {\( s_2 \)};

\end{tikzpicture}
\qquad
\begin{tikzpicture}[every node/.style={font=\small}, thread/.style={thick}, invalid/.style={dashed},]

\node (main) {f};
\node (t1) [right=of main] {g};
\node (t2) [right=of t1] {h};

\draw[thread] (main.south) -- ++(0,-4);
\draw[thread] (t1.south) -- ++(0,-4);
\draw[thread] (t2.south) -- ++(0,-4);

\coordinate (m1) at ($(main.south)+(0,-0.5)$);
\coordinate (m2) at ($(main.south)+(0,-2.5)$);
\coordinate (m3) at ($(main.south)+(0,-3.5)$);

\coordinate (t21) at ($(t2.south)+(0,-1.5)$);
\coordinate (t22) at ($(t2.south)+(0,-2)$);

\coordinate (t11) at ($(t1.south)+(0,-3)$);

\draw[->,invalid] (m1) -- (t21);
\node[left] at (m1) {\( \mathit{ack}_s \)};
\node[right] at (t21) {\( \mathit{ack}_r \)};

\draw[<-,invalid] (m2) -- node[above right] {20} (t22);
\node[left] at (m2) {\textcolor{gray}{\( v_1\)} \( r \)};
\node[right] at (t22) {\( s_2 \)};

\draw[<-] (m3) -- node[below] {10} (t11);
\node[left] at (m3) {\textcolor{gray}{\( v_2\)} \( r \)};
\node[right] at (t11) {\( s_1 \)};

\end{tikzpicture}

%% file: sections/ProofOfProperties.tex
\section{Proof of Soundness and Completeness}
\label{sec:proof-of-properties}
As briefly explained, the proof of soundness and completeness is given by a mutual construction between (i) a resolution proof of the unsatisfiability and (ii) a reduction sequence of the program leading to an assertion failure.
In this section, we first define what a resolution proof is, and then define yet another operational semantics that helps us relate a proof to a reduction sequence.
Then we give the actual constructions.

\subsection{Constrained SLD Resolution}
It is well known that SLD resolution is refutationally complete for Horn clauses, and we use SLD resolution steps as the proof-side counterpart of an execution trace.
Here we review constrained SLD resolution.
A goal \( G = \Goal{\seq{\Atom}} \) is a sequence of ground atomic formulas. We write $\Goal{\seq{\Atom}} \SLDstep{\CHCSystem} \Goal{\seq{\Atom'}}$ to denote that the goal $\Goal{\seq{\Atom}}$ takes one resolution step to goal $\Goal{\seq{\Atom'}}$ under the entire system of CHCs $\CHCSystem = (\CHCSet, \sortenv)$. The resolution rule is defined as follows.

\infrule[Constrained SLD Resolution]
  {\Head \leftarrow \Atom'_1 \land \ldots \land \Atom'_m \land \Constraint \in \CHCSet \andalso \exists \subst.\ \Atom_i = \Head \subst \andalso \models \Constraint \subst}
  {\Goal{\Atom_1, \ldots, \Atom_i, \ldots, \Atom_n} \SLDstep{\CHCSystem} \Goal{\Atom_1, \ldots, \Atom'_1 \subst, \ldots, \Atom'_m \subst, \ldots, \Atom_n}}

\subsection{Extended Configuration}
As a mediator between configurations of the operational semantics (defined in Figure~\ref{fig:target-operational-semantics}) and goals of SLD resolution, we introduce extended configurations.
The idea is to let the configurations carry the timestamps, prophecies, and the failure flag, and only allow them to transition if the guesses are correct.

\begin{remark}[Abuse of notation]
  In what follows, we use \( b \) and \( t \) not only to denote integer variables, but also, by abuse of notation, to denote integer constants.
  We may also use \( \SORTED \) as an interpreted predicate.
  \qed
\end{remark}

\emph{Extended configurations} are given by the following grammar.
\begin{align*}
  &\mbox{(extended configurations) }  \EConf ::= (\EThreads, \EQu, \TSMap, \ProphecyMap, \FailureMap) \mid \FAIL\\
  &\mbox {(thread pools with identifiers) } \EThreads ::= \emptyset \mid \EThreads \cup \{ (i, f(\seq a)) \} \\
  &\mbox{(extended channel context) } \EQu ::= \emptyset \mid \EQu \uplus \SET{ r \mapsto (S, \Eqc, t) }\\
  &\mbox{(extended message queue) }  \Eqc ::= [] \mid (t, n) :: \Eqc \\
  &\mbox{(map from thread ids to timestamps) } \TSMap  ::= \SET{ i \mapsto t_i \mid i \in \TID}\\
  &\mbox{(map from senders and receivers to their prophecies) } \ProphecyMap  ::= \emptyset \mid \ProphecyMap \uplus \SET{ x \mapsto \Eqc }\\
  &\mbox{(map from thread ids to thread outcomes) }\FailureMap  ::= \SET{ i \mapsto b_i \mid i \in \TID, b_i = \TRUE \lor b_i = \FALSE }
  \end{align*}
  Now each thread is given a thread id \( i \), which is an element of a denumerable set \( \TID \), and the message queue contains not only the value that was sent but also the time it was sent.
  For technical convenience, the extended channel context also holds the time for the most recent send to a receiver \( r \), which is the third element of \( \EQu(r)\).
  The map \( U \) holds the last time each thread \( i \) performed a send or receive operation.
  The maps \( \ProphecyMap \) and \( \FailureMap \) guess the future of each thread \( i \).

\subparagraph*{Operational semantics}
The operational semantics of extended configurations is defined in \autoref{fig:operational-semantics-extended}.
Here $\MergeT$ is a function that merges its arguments while preserving the order within each argument list and arranging the result in ascending order of timestamps.
Note that we have premises akin to the constraints in the translated CHCs.
For example, in \rn{EStep-New} we require the new prophecies to be sorted.

\begin{figure}[tpb]
  \fbox{\( (\thread, \EQu, \TSMap, \ProphecyMap, \FailureMap) \stepC{D} (\thread', \EQu', \TSMap', \ProphecyMap', \FailureMap')\)}
\infrule[EStep-If-True]
  {f(\seq{x}:\seq{\bt}) = \ifstmt{a_0}{g_1(\seq{a}_1)}{g_2(\seq{a}_2}) \andalso [\seq{a}/\seq{x}]a_0 \eval n \andalso n \ne 0}
  {((i, f(\seq{a})), \EQu, \TSMap, \ProphecyMap, \FailureMap) \stepC{D} ((i, [\seq{a}/\seq{x}]g_1(\seq{a}_1)), \EQu, \TSMap, \ProphecyMap, \FailureMap)}
\infrule[EStep-If-False]
  {f(\seq{x}:\seq{\bt}) = \ifstmt{a_0}{g_1(\seq{a}_1)}{g_2(\seq{a}_2}) \andalso [\seq{a}/\seq{x}]a_0 \eval n \andalso n = 0}
  {((i, f(\seq{a})), \EQu, \TSMap, \ProphecyMap, \FailureMap) \stepC{D} ((i, [\seq{a}/\seq{x}]g_2(\seq{a})), \EQu, \TSMap, \ProphecyMap, \FailureMap)}
\infrule[EStep-Func]
  {f(\seq{x}:\seq{\bt}) = g(\seq{a}') \in D}
  {((i, f(\seq{a})), \EQu, \TSMap, \ProphecyMap, \FailureMap) \stepC{D} ((i, [\seq{a}/\seq{x}]g(\seq{a}')), \EQu, \TSMap, \ProphecyMap, \FailureMap)}
\infrule[EStep-New]
  {f(\seq{x}:\seq{\bt}) = \newstmt{s}{r}{g(\seq{a}')} \in D \andalso \SORTED(\ProphecyMap(s')) \andalso \ProphecyMap(s') = \ProphecyMap(r') \andalso s', r' \FRESH}
  {((i, f(\seq{a})), \EQu, \TSMap, \ProphecyMap \setminus \SET{s', r'}, \FailureMap) \stepC{D} ((i, [s'/s, r'/r][\seq{a}/\seq{x}]g(\seq{a}')), \EQu[ r' \mapsto (\SET{s'},  \NIL, -\infty) ], \TSMap, \ProphecyMap, \FailureMap)}
\infrule[EStep-Send]
  {f(\seq{x}:\seq{\bt}) = \sendstmt{a_0}{s}{g(\seq{a'})}\in D \andalso [\seq{a}/\seq{x}]a_0 \eval n \andalso \EQu(r) = (S, \Eqc, t_S) \\
  s' = [\seq{a}/\seq{x}] s \andalso s' \in S \andalso t' \ge \TSMap(i) \andalso t' > t_S  \andalso t' < \MinTimestamp{\ProphecyMap(s'')} \text{ for every } s'' \in S}
  {((i, f(\seq{a})), \EQu, \TSMap, \ProphecyMap[ s' \mapsto (t', n) :: \ProphecyMap(s') ], \FailureMap) \\\stepC{D} ((i, [\seq{a}/\seq{x}]g(\seq{a}')), \EQu[ r \mapsto (S, \Eqc \cdot [(t', n)], t') ], \TSMap[ i \mapsto t' ], \ProphecyMap, \FailureMap)}
  \infrule[EStep-Recv]
  {f(\seq{x}:\seq{\bt}) = \recvstmt{y}{r}{g(\seq{a}')} \in D \andalso t'' \ge \TSMap(i) \andalso t'' > t'  \\ r' = [\seq{a}/\seq{x}]r \andalso \EQu(r') = (S, (t', n)::\Eqc, t_S)}
  {((i, f(\seq{a})), \EQu, \TSMap, \ProphecyMap[ r' \mapsto (t', n) :: \ProphecyMap(r') ], \FailureMap) \\\stepC{D} ((i, [n/y][\seq{a}/\seq{x}]g(\seq{a}')), \EQu[ r' \mapsto (S, \Eqc, t_S) ], \TSMap[ i \mapsto t'' ], \ProphecyMap, \FailureMap)}
\infrule[EStep-Dup]
  {f(\seq{x}:\seq{\bt}) = \dupstmt{y}{s}{g(\seq{a}')} \in D \andalso s'' \FRESH \andalso s' = [\seq{a}/\seq{x}]s \andalso s' \in S \andalso \EQu(r) = (S, \Eqc, t_S)}
  {((i, f(\seq{a})), \EQu, \TSMap, \ProphecyMap[ s' \mapsto \MergeT(\ProphecyMap(s'), \ProphecyMap(s'')) ] \setminus \SET{s''}, \FailureMap)\\
  \stepC{D} ((i, [s''/y][\seq{a}/\seq{x}]g(\seq{a}')), \EQu[ r \mapsto (S \cup \SET{s''}, \Eqc, t_S) ], \TSMap, \ProphecyMap, \FailureMap)}
\fbox{\( (\EThreads, \EQu, \TSMap, \ProphecyMap, \FailureMap) \stepC{D} (\Threads', \EQu', \TSMap', \ProphecyMap', \FailureMap')\)}
  \infrule[EStep-Conc]
  {(\thread, \EQu, \TSMap, \ProphecyMap, \FailureMap) \stepC{D} (\thread', \EQu', \TSMap', \ProphecyMap', \FailureMap')}{(\Threads \cup \{ \thread \}, \EQu, \TSMap, \ProphecyMap, \FailureMap) \stepC{D} (\Threads \cup \{ \thread'\}, \EQu', \TSMap', \ProphecyMap', \FailureMap') }
\infrule[EStep-Spawn]
  {f(\seq{x}:\seq{\bt}) = \spawnstmt{g_1(\seq{a}_1)}{g_2(\seq{a}_2)} \in D \andalso j \FRESH\ }
  {(\Threads \cup \SET{(i, f(\seq{a})}), \EQu, \TSMap, \ProphecyMap, \FailureMap[ i \mapsto (\FailureMap(i) \lor \FailureMap(j)) ] \setminus \SET{j}) \\\stepC{D} (\Threads \cup \SET{(j, [\seq{a}/\seq{x}]g_1(\seq{a}_1)), (i, [\seq{a}/\seq{x}]g_2(\seq{a}_2))}, \EQu, \TSMap[ j \mapsto \TSMap(i) ], \ProphecyMap, \FailureMap)}
\infrule[EStep-Fail]
  {f(\seq{x}:\seq{\bt}) = \FAIL \in D \andalso \ProphecyMap = \SET{ x \mapsto \NIL \mid x \in \OutChannels{\EQu}}\\
  \FailureMap = \SET{ i \mapsto \TRUE } \uplus \SET{ j \mapsto \FALSE \mid (j, g(\seq{x}:\seq{\bt})) \in \Threads, g(\seq{x}:\seq{\bt}) \ne \FAIL} \\
  \uplus\ \SET{j \mapsto b \mid (j, g(\seq{x}:\seq{\bt})) \in \Threads, g(\seq{x}:\seq{\bt}) = \FAIL \in D, b = \TRUE \lor b = \FALSE}}
  {(\Threads \cup \SET{(i, f(\seq{a}))}, \EQu, \TSMap, \ProphecyMap, \FailureMap) \stepC{D} \FAIL}
  \caption{Operational Semantics of Extended Configuration}
\label{fig:operational-semantics-extended}
\end{figure}

We explain how the content of \( \EConf \) is meant to evolve and what invariants will be maintained.
\begin{definition}
  \label{def:well-formedness-of-extended-configurations}
  An extended configuration $\EConf = (\Threads, \EQu, \TSMap, \ProphecyMap, \FailureMap)$ is well-formed if the following conditions hold:
\begin{enumerate}
  \item For each sender $s$, $s$ occurs in $\Threads$ if and only if $s \in \DOM{ \ProphecyMap }$ and there exists a unique receiver $r$ with $\EQu(r) = (S, \Eqc, t_S)$ and $s \in S$.
  \item For each receiver $r$, $r$ occurs in $\Threads$ if and only if $r \in \DOM{\ProphecyMap}$ and $r \in \DOM{ \EQu }$.
  \item For each receiver $r$ such that $\EQu(r) = (\SET{s_1, \ldots, s_n}, \Eqc, t_S)$, the followings hold:
  \begin{align*}
    &\SORTED(\Eqc), \\
    &\SORTED(\ProphecyMap(s_i)) \mbox{ for each } i = 1, \ldots, n,\\
    &\SORTED(\MergeT(\Eqc, \ProphecyMap(s_1), \ldots, \ProphecyMap(s_n))),\\
    &\ProphecyMap(r) =  \MergeT(\Eqc, \ProphecyMap(s_1), \ldots, \ProphecyMap(s_n)), \\
    &\MaxTimestamp{ \Eqc } \le t_S, \mbox{ and }\\
    &t_S < \MinTimestamp{ \ProphecyMap(s_i) } \mbox{ for each } i = 1, \ldots, n.
  \end{align*}
  \item For each thread identifier $i$, $i$ occurs in $\Threads$ if and only if $i \in \DOM{\TSMap}$ and $i \in \DOM{\FailureMap}$.
\end{enumerate}
\qed
\end{definition}
Well-formedness is preserved by reduction.
\begin{lemma}[Preservation of well-formedness]
  \label{lem:well-formedness-invariant}
  If \( \EConf_1 \) is well-formed and \( \EConf_1 \stepC{D}\EConf_2 \), then so is \( \EConf_2 \).
\end{lemma}
\begin{proof}
We proceed by case analysis on the rule used for \( \EConf_1 \Longrightarrow_D \EConf_2 \).
We only prove the interesting cases.
In all the cases below, we let  \( \EConf_1 = (\EThreads,\EQu,\TSMap,\ProphecyMap,\FailureMap) \).
\begin{itemize}
  \item Case \rn{EStep-New}

    The step must be of the form
   \begin{align*}
     &((i,f(\vec a)), \EQu, \TSMap, \ProphecyMap' \setminus \{ s',r' \}, \FailureMap) \\
     &\qquad \stepC{D}
       ((i,[s'/s,r'/r][\seq a/\seq x]g(\vec a')), \EQu[r'\mapsto(\{s'\},[], -\infty)], \TSMap,\ProphecyMap',\FailureMap),
   \end{align*}
    where \(s'\) and \(r'\) are fresh.
    We now check the third condition of well-formedness.
    Since \( r' \) part is the only update part of  \( \EQu \), it suffices to check that
    \begin{enumerate}
      \item \( \ProphecyMap'(r')=\MergeT([],\ProphecyMap'(s'))=\ProphecyMap'(s') \)
      \item \( \SORTED(\ProphecyMap'(s')) \) and
      \item   \( -\infty < \MinTimestamp{\ProphecyMap(s')} \).
    \end{enumerate}
    The first two conditions are exactly the premise of \rn{E-New}, and the third condition is trivial.
    This also ensures the second condition of well-formedness, saying that \( \EQu(r') \) and \(M(r')\) are defined.

    Freshness gives the required uniqueness of the receiver associated with \(s'\) (i.e.~the first condition of well-formedness).
    The fourth condition of well-formedness also holds because \( (\TSMap, \FailureMap) \) are unchanged and no threads are generated or killed.
    Hence \( \EConf_2 \) is well-formed.

  \item Case \rn{EStep-Send}

    We must have
    \begin{align*}
      &((i, f(\seq{a})), \EQu, \TSMap, \ProphecyMap'[ s' \mapsto (t', n) :: \ProphecyMap'(s') ], \FailureMap) \\
      &\qquad \stepC{D} ((i, [\seq{a}/\seq{x}]g(\seq{a}')), \EQu[ r \mapsto (S, \Eqc \cdot [(t', n)], t') ], \TSMap[ i \mapsto t' ], \ProphecyMap', \FailureMap)
    \end{align*}
    with
    \begin{align}
     & \EQu(r) = (S, \Eqc, t_S) \nonumber  \\
     &s' = [\seq{a}/\seq{x}] s \andalso s' \in S \nonumber \\
     &t' \ge \TSMap(i) \andalso t' > t_S   \label{eq:lem:well-formedness-invariant:send:ts} \\
     &t' < \MinTimestamp{\ProphecyMap(s_i)} \text{ for every } 1 \le i \le m \label{eq:lem:well-formedness-invariant:send:min}
    \end{align}
    where \(   S= \{s', s_1, \ldots, s_m\} \).

    It suffices to check the third condition of the well-formedness rule against \( r \) for \( \EConf_2 \); the rest follows from that of \( \EConf_1 \).

    Since \(\EConf_1\) is well-formed, we have
    \[
    \ProphecyMap(r)  = \MergeT(\Eqc, (t', n)::\ProphecyMap'(s'), \ProphecyMap'(s_1),\ldots,\ProphecyMap'(s_m)).
    \]
    Thus, we have
    \begin{align*}
      \ProphecyMap'(r)
      &= \MergeT(\Eqc, (t', n)::\ProphecyMap'(s'), \ProphecyMap'(s_1),\ldots,\ProphecyMap'(s_m))  \\
      &= \MergeT(\Eqc \cdot [(t',n)], \ProphecyMap'(s'), \ProphecyMap'(s_1),\ldots,\ProphecyMap'(s_m)) \tag{by \eqref{eq:lem:well-formedness-invariant:send:min}}.
    \end{align*}

    The sortedness of \( \Eqc \cdot[(t',n)] \) follows from the sortedness of \( \Eqc \) and the following inequality
    \begin{align*}
      \MaxTimestamp{\Eqc} &\leq t_S \tag{by well-formedness of \( \EConf_1\)} \\
      &< t' \tag{by \eqref{eq:lem:well-formedness-invariant:send:ts}}
    \end{align*}
    and the sortedness of \(\ProphecyMap'(s')\) follows the sortedness of \(\ProphecyMap(s')=(t',n)::\ProphecyMap'(s')\).

  \item Case \rn{EStep-Recv}

    Similar to the previous case.

  \item Case \rn{EStep-Dup}
    In this case, we must have
    \begin{align*}
      &((i, f(\seq{a})), \EQu, \TSMap, \ProphecyMap'[ s' \mapsto \MergeT(\ProphecyMap'(s'), \ProphecyMap'(s'')) ] \setminus \SET{s''}, \FailureMap) \\
        &\qquad \stepC{D} ((i, [s''/y][\seq{a}/\seq{x}]g(\seq{a}')), \EQu[ r \mapsto (S \cup \SET{s''}, \Eqc, t_S) ], \TSMap, \ProphecyMap', \FailureMap)
    \end{align*}
    with
    \begin{align*}
      & \ProphecyMap = \ProphecyMap'[ s' \mapsto \MergeT(\ProphecyMap'(s'), \ProphecyMap'(s'')) ] \setminus \SET{s''} \\
      &\EQu(r) = (S, \Eqc, t_S) \\
      &s' = [\seq{a}/\seq{x}]s \text{ and } s' \in S.
    \end{align*}

    The first condition of the well-formedness holds because \( s'' \) is the only added sender, and this is related to \( r \) as the channel context is updated to \( \EQu[ r \mapsto (S \cup \SET{s''}, \Eqc, t_S) ] \).

    Now we check the third condition of the well-formedness.
    The channels that we need to consider are \( r \), \( s' \) and \( s'' \).
    Suppose that  \(  S=\{s',s_1, \ldots, s_m\} \).
    Since \(\EConf_1\) is well-formed, we have
    \[
    \ProphecyMap(r) = \MergeT(\Eqc,\ProphecyMap(s'), \ProphecyMap'(s_1),\ldots,\ProphecyMap'(s_m)).
    \]
    Therefore
    \begin{align*}
      \ProphecyMap'(r)
      &= \ProphecyMap(r)\\
      &= \MergeT(\Eqc,\ProphecyMap(s'), \ProphecyMap(s_1),\ldots,\ProphecyMap(s_m))  \tag{by the above equation}\\
      &= \MergeT(\Eqc,\ProphecyMap'(s'),\ProphecyMap'(s''), \ProphecyMap'(s_1),\ldots,\ProphecyMap'(s_m)) \tag{since \( \ProphecyMap(s') \) is the merge of \( \ProphecyMap'(s') \) and \( \ProphecyMap'(s'') \)}.
    \end{align*}
    Sortedness of \( \ProphecyMap'(s') \) and \( \ProphecyMap'(s'') \) follows from sortedness of \(\ProphecyMap(s')\) and the fact that merge preserves the order of the two
    sublists.
    The remaining conditions to check are that for \( t' \), namely the conditions on the last time the send was performed.
    The only non-trivial conditions to check are
    \begin{align*}
      t_S< \MinTimestamp{\ProphecyMap'(s')}  \quad \text{ and } \quad t_S < \MinTimestamp{\ProphecyMap'(s'')}.
    \end{align*}
    These inequalities hold because
    \begin{align*}
      t &<  \MinTimestamp{\ProphecyMap(s')} \tag{by the well-formedness of \( \EConf_1 \)} \\
      &\le \MinTimestamp{\ProphecyMap'(s')} \tag{since \( \ProphecyMap'(s') \) is a sublist of \( \ProphecyMap(s') \)}
    \end{align*}
    and similarly for \( \ProphecyMap'(s'')\).

    The preservation of the second and the fourth conditions of well-formedness is trivial.
\end{itemize}
\end{proof}

\subparagraph*{Translation}
Since extended configurations carry information about the failure flags, timestamps and prophecies, we can translate them into a sequence of ground atomic formulas, which we show to correspond to a goal in a resolution proof.
The translation from extended configurations to atomic formulas is defined as follows.
\begin{align*}
  &\transmap {((i, f(\seq{a})), \EQu, \TSMap, \ProphecyMap, \FailureMap)} \defeq F(\FailureMap(i), \TSMap(i), \subst_\ProphecyMap \seq{a}) \\
  &\transmap{\FAIL} \defeq \emptyset \\
  &\transmap{(\SET{ \thread_1, \ldots, \thread_n }, \EQu, \TSMap, \ProphecyMap, \FailureMap)} \defeq \Atom_1, \ldots, \Atom_n \qquad \text{ where } \Atom_i = \transmap{(\thread_i, \EQu, \TSMap, \ProphecyMap, \FailureMap)}
\end{align*}

Here \( \subst_\ProphecyMap \) is the map \( \ProphecyMap \) regarded as a substitution, namely, a map that associates \( x \) with the constant representing the list \( \ProphecyMap(x) \).
Since we will consider well-formed configurations for the rest of this section, \( \subst_\ProphecyMap \) will close all the list variables.
Therefore, the translation of a well-formed configuration is a sequence of ground atomic formulas.

In what follows, we shall also write \( \transmap {D} \) for the system of CHCs obtained by translating the set of function definitions \( D \).
In other words, \( \transmap{D} \) is \( \transmap{(D, f(\seq a))}\) except for the goal clause corresponding to \( f(\seq a) \).

\subsection{Proof of Soundness}
We now prove the only if direction of \autoref{thm:soundness}.
We first show that reduction steps can be simulated by the resolution steps.
The following lemma considers all the cases except for the case where the configuration reduces to \( \FAIL \).
\begin{lemma}[Simulation of extended reduction]
  \label{lem:simulation-without-estep-fail}
  Let $D$ be a set of function definitions such that $\transmap{D} = \CHCSystem$. Let $\EConf_1$ and $\EConf_2$ be extended configurations such that $\EConf_1 \stepC{D} \EConf_2$, \( \EConf_1 \) is well-formed and \( \EConf_2 \neq \FAIL \).
  If $\transmap{\EConf_i} =  \seq{A}_i$ (for $i = 1, 2$), then $\Goal{\seq{A}_1} \SLDstep{\CHCSystem} \Goal{\seq{A}_2}$.
\end{lemma}
\begin{proof}

  We proceed by case analysis on the evaluation rule. 
  
  \begin{itemize}
    \item Case \rn{EStep-If-True}. 

      In this case, we must have
      \begin{align}
      \EConf'_1 &= ((i, f(\seq{a})), \EQu, \TSMap, \ProphecyMap, \FailureMap) \nonumber \\
      \EConf'_2 &= ((i, [\seq{a}/\seq{x}]g_1(\seq{a}_1)), \EQu, \TSMap, \ProphecyMap, \FailureMap) \nonumber \\
      f(\seq{x}:\seq{\bt}) &= \ifstmt{a_0}{g_1(\seq{a}_1)}{g_2(\seq{a}_2}) \in D \nonumber \\
      [\seq{a}/\seq{x}]a_0 &\eval n \text{ with } n \ne 0. \label{eq:lem:simulation:if-true}
      \end{align}
      The configuration $\EConf'_1$ is translated into $F(\FailureMap(i), \TSMap(i), \seq{a} \subst_\ProphecyMap)$ and the configuration $\EConf'_2$ is translated into $G_1(\FailureMap(i), \TSMap(i), \seq{a}_1 \subst_\ProphecyMap)$.
      Moreover, $\CHCSystem$ contains a CHC
      \[
      F(b, t, \seq{x}) \Leftarrow a_0 \ne 0 \land G_1(b, t, \seq{a}_1).
      \]
      We can apply this clause to the resolution step.
      We are left to show \( \models [\seq{a}/\seq{x}]a_0 \neq 0 \), and this follows from \eqref{eq:lem:simulation:if-true}.

    \item Cases \rn{EStep-If-False} and \rn{EStep-Func}.

      Similar to the previous case.

\item Case \rn{EStep-New}.

  It suffices to show that $\ProphecyMap(r') =  \ProphecyMap(s')$ and $\SORTED(\ProphecyMap(s'))$. By the premise of \rn{EStep-New} this indeed holds.

\item Case \rn{EStep-Send}.

      It must be the case that
      \begin{align*}
        \EConf'_1 &= ((i, f(\seq{a})), \EQu, \TSMap,  \ProphecyMap',\FailureMap) \\
        \EConf'_2  &= ((i, [\seq{a}/\seq{x}]g(\seq{a}')), \EQu[ r \mapsto (S, \Eqc \cdot [(t', n)], t') ], \TSMap[ i \mapsto t' ], \ProphecyMap, \FailureMap)
      \end{align*}
      with
      \begin{align}
        \ProphecyMap' &= \ProphecyMap[s' \mapsto (t', n) :: \ProphecyMap(s')]  \label{eq:lem:simulation:send:ProphecyMap} \\
        \EQu(r) &= (S, \Eqc, t_S) \\
        &f(\seq{x}:\seq{\bt}) = \sendstmt{a_0}{s}{g(\seq{a'})}\in D \nonumber \\
        &[\seq{a}/\seq{x}]a_0 \eval n  \label{eq:lem:simulation:send:azero} \\
        s' &= [\seq{a}/\seq{x}] s \andalso s' \in S \label{eq:lem:simulation:send:sprime}\\
        \quad &t' \ge \TSMap(i) \andalso t' > t_S \label{eq:lem:simulation:send:time}
      \end{align}

      The configuration $\EConf'_1$ is translated into \( \Atom_1 = F(\FailureMap(i), \TSMap(i), \subst_{\ProphecyMap'}\seq{a})  \) and the configuration $\EConf'_2$ is translated into $\Atom_2 = G(\FailureMap(i), t', \subst_\ProphecyMap[\seq{a}/\seq{x}]\seq{a}') $. Moreover, $\CHCSystem$ contains a CHC
      \[
      F(b, t, \seq{x}) \Leftarrow G(b, t', [s''/s] \seq{a}') \land s = (t', a_0) :: s'' \land t < t'.
      \]
      We show that we can deduce \( \Atom_1 \SLDstep{\CHCSystem} \Atom_2 \) by applying this clause together with a substitution \( \subst = \subst_{\ProphecyMap'} [\TSMap(i)/t][\FailureMap(i)/b][\seq a/\seq x]\).
      That is, we show that the constraint
      \begin{align*}
         \subst_{\ProphecyMap'}[\seq a/\seq x] [U(i)/t](s = (t', a_0) :: s'' \land t < t')
      \end{align*}
      is valid and
      \[
      \subst_{\ProphecyMap'}[\FailureMap(i)/b][\seq a/\seq x]  G(b, t', [s''/s] \seq{a}') \iff G(\FailureMap(i), t', \subst_\ProphecyMap[\seq{a}/\seq{x}]\seq{a}')
      \]
      modulo the constraint.
      The former holds because the equality constraint is valid, as demonstrated as follows
      \begin{align*}
        &\subst_{\ProphecyMap'} [\seq a/\seq x](s = (t', a_0) :: s'') \\
        &\equiv \subst_{\ProphecyMap'}s' = (t',  \subst_{\ProphecyMap'}[\seq a/\seq x]a_0) ::  \subst_{\ProphecyMap'}[\seq a/\seq x] s' \tag{by \eqref{eq:lem:simulation:send:sprime}}\\
        &\equiv (t', n) :: M(s') = (t', \subst_{\ProphecyMap'}[\seq a/\seq x]a_0) :: M(s')\tag{by  \eqref{eq:lem:simulation:send:ProphecyMap} and \(s'' \notin \seq x\) }\\
        &\iff (t', n) :: M(s') = (t', n) :: M(s') \tag{by \eqref{eq:lem:simulation:send:azero}}
      \end{align*}
      and the inequality follows from~\eqref{eq:lem:simulation:send:time}.
      The latter also holds because
      \begin{align*}
      \subst_{\ProphecyMap'}[\seq a/\seq x][s''/s] \seq{a}'
      &\iff \subst_{\ProphecyMap'}[\seq{a}/\seq{x}][M(s')/s]\seq{a}' \tag{by the above equality}\\
      &\equiv \subst_\ProphecyMap  [\seq{a}/\seq{x}]\seq{a}' \tag{by \eqref{eq:lem:simulation:send:ProphecyMap} and \eqref{eq:lem:simulation:send:sprime}}.
      \end{align*}
    \item Cases \rn{EStep-Recv} and \rn{EStep-Dup}

      Similar to the previous case. %

\item Case \rn{EStep-Conc}.

      In this case, the configurations must be of the shape
      \begin{align*}
        \EConf_1 &= (\Threads \cup \SET{ \thread }, \EQu, \TSMap, \ProphecyMap, \FailureMap) \text{ and} \\
        \EConf_2 &= (\Threads \cup \SET{ \thread' }, \EQu', \TSMap', \ProphecyMap', \FailureMap').
      \end{align*}
      Since executing thread $T$ changes only the prophecies of that channels owned by $T$ in $\ProphecyMap$, and only the timestamp of $T$ in $\TSMap$, the configurations $(\Threads, \EQu, \TSMap, \ProphecyMap, \FailureMap)$ and $(\Threads, \EQu', \TSMap', \ProphecyMap', \FailureMap')$ are translated into the same goal. By the above cases, for atomic formulas $A$ translated from the configuration of $\thread$ and $A'$ translated from the configuration of $\thread'$, we have $\Goal{A} \SLDstep{\CHCSystem} \Goal{A'}$.

    \item Case \rn{EStep-Spawn}.

      Similar to the case \rn{EStep-Conc}.
\end{itemize}
\end{proof}

Now we prove the case for an assertion failure.
On the CHC side, this corresponds to discharging all the goals.
\begin{lemma}\label{lem:simulation-estep-fail}
  Let $D$ be a set of function definitions such that $\transmap{D} =  \CHCSystem$. Let $\EConf$ be an extended configuration such that $\EConf \stepC{D} \FAIL$.
  If $\transmap {\EConf} = \seq{A}_i$, then $\Goal{\seq{A}} \mSLDstep{\CHCSystem} \Goal{}$.
\end{lemma}
\begin{proof}
  Let $\EConf = (\Threads, \EQu, \TSMap, \ProphecyMap, \FailureMap)$.
  Since the last rule applied must be \rn{EStep-Fail}, we must have
  \begin{align}
    \Threads &= \Threads' \cup \{ (i, f(\seq a) \} \text{ with } f(\seq{x}:\seq{\bt}) = \FAIL \in D \nonumber \\
    \ProphecyMap &= \SET{ x \mapsto \NIL \mid x \in \OutChannels{\EQu}} \label{eq:lem:simulation-step-fail:prophecy} \\
    \FailureMap& = \SET{ i \mapsto \TRUE }
                 \begin{aligned}[t]
                 &\uplus \SET{ j \mapsto \FALSE \mid (j, g(\seq{x}:\seq{\bt})) \in \Threads, g(\seq{x}:\seq{\bt}) \ne \FAIL} \\
                 &\uplus\ \SET{j \mapsto b \mid (j, g(\seq{x}:\seq{\bt})) \in \Threads, g(\seq{x}:\seq{\bt}) = \FAIL \in D, b = \TRUE \lor b = \FALSE}.  \label{eq:lem:simulation-step-fail:outcome}
                \end{aligned}
  \end{align}

  It suffices to show that $\Goal{ A_k } \mSLDstep{\CHCSystem} \Goal{}$ for each $A_k \in \seq{A}$, where each \( A_k \)
  corresponds to a thread \( (k, f_k(\seq a_k))\).
  For each $(k, f_k(\seq{a})) \in \Threads$, the configuration $((k, f_k(\seq{a}_k)), \EQu, \TSMap, \ProphecyMap, \FailureMap)$ is translated into $F_k(\FailureMap(k), \TSMap(k), \seq{a}_k \subst_\ProphecyMap)$.

  We distinguish two cases depending on whether $\FailureMap(k)$ is $\TRUE$ or $\FALSE$.

  If $\FailureMap(k) = \TRUE$, then by~\eqref{eq:lem:simulation-step-fail:outcome}, we have $f_k(\seq{x}:\seq{\bt}) = \FAIL \in D$. Hence $\CHCSystem$ contains the CHC
  \[
    F_k(b, t, \seq{x}) \Leftarrow b = \TRUE \land \seq{s} = \seq{\NIL}$ where $\seq{s} = \OutChannels{ \seq x_k }.
  \]
  By applying this clause we get $\Goal{F_k(\FailureMap(k), \TSMap(k), \seq{a} \subst_\ProphecyMap)} \SLDstep{\CHCSystem} \Goal{}$; the constraints are valid because of~\eqref{eq:lem:simulation-step-fail:prophecy}.

  Now we consider the case $\FailureMap(k) = \FALSE$.
  By construction, $\CHCSystem$ contains the CHC
  \[
    F_k(b, t, \seq{x}) \Leftarrow b = \FALSE \land \seq{s} = \seq{\NIL}$ where $\seq{s} = \OutChannels{ \seq x_k }.
  \]
    Again, by applying this clause and using~\eqref{eq:lem:simulation-step-fail:prophecy}, we conclude $\Goal{F_k(\FailureMap(k), \TSMap(k),  \subst_\ProphecyMap \seq{a})} \SLDstep{\CHCSystem} \Goal{}$.
\end{proof}

Although we have been working with a run of extended operational semantics, we can relate it to that of the normal operational semantics.
Since the values that are sent and received in the future can be obtained by actually running the program, we have the following:
\begin{lemma}
  \label{lem:efail-construct}
  Let $(D, f(\seq{a}))$ be a program and suppose that \( (\{ f(\seq a) \}, \emptyset) \mstepC{D} \FAIL \).
  Then we have $(\SET{(i, f(\seq{a}))}, \emptyset, \SET{i \mapsto 0}, \emptyset, \SET{i \mapsto \TRUE}) \mstepC{D} \FAIL$.
\end{lemma}
\begin{proof}[(Proof Sketch)]
\newcommand*{\forget}[1]{|#1|}
\newcommand*{\exttime}{\mathsf{time}}
\newcommand*{\annot}{\mathsf{annot}}
We briefly explain the construction.
Let
\[
  \pi \defeq
  C_0 \stepC{D} C_1 \stepC{D} \cdots \stepC{D} C_m=\FAIL
\]
be a failing trace.
We write \(C_k =(\thread_k, \Qu_k)\) for \(k < m\).
Suppose that threads in \(\pi\) have been annotated with identifiers, and write \(\widehat \Threads_k\) for the thread pool at position \(k\).

Let \(E(\pi)\) be the list of all send and receive events in \(\pi\), ordered by
their occurrence in the trace: \(  E(\pi)=e_1,\ldots,e_N \).
Here we will not be specific about the shape of events; for example, it could be defined as \( e ::= \mathsf{send}(i, k, s, n) \mid \mathsf{recv}(i, k, r, n,s) \) with the information of thread identifier, the position of the execution trace the event happened, the channels used for send/receive, and the value sent/received.
We also define \( \exttime(e_p) \defeq p \).

For each position \(k\), we define
\[
  \EConf_k^\pi \defeq (\widehat T_k,\EQu_k,\TSMap_k,\ProphecyMap_k,\FailureMap_k).
\]

The components are defined as follows.
First, for each thread identifier \( i \),
\[
  \TSMap_k(i) \defeq  \max \left( \{0\}\cup \{\,\exttime(e) \mid e \text{ is a send or receive event performed by } i
  \text{ before position } k\,\}\right).
\]
The map \(\EQu_k\) is obtained from \(Q_k\) by timestamping queued messages.
That is, if \( Q_k(r)=(S,n_1\cdots n_n), \) then \(  \EQu_k(r)=(S,(t_1,n_1)\cdots(t_n,n_n),t) \)
where \(t_j\) is the timestamp of the send event that produced the queued
message \(n_j\), and
\[
  t =
  \max\bigl(\{0\}\cup
  \{\,\exttime(e)\mid e \text{ is a send event to } r
       \text{ before position } k\,\}\bigr).
\]
Now we define \( \ProphecyMap_k \) for each sender \(s\),
\[
  \ProphecyMap_k(s) \defeq [(\exttime(e_1),n_1); \ldots ;(\exttime(e_\ell),n_\ell)],
\]
where \(e_1,\ldots,e_\ell\) are exactly the send events through \(s\) occurring at or after position \(k\), listed in trace order, and \(n_j\) is the value sent by \(e_j\).
Note that we have \( \SORTED(  \ProphecyMap_k(s)) \).
For each receiver \( r \), if \( \EQu_k(r)=(S, \Eqc ,t) \), define \( \ProphecyMap_k(r) \defeq \mathit{Merge}'(\Eqc,\{\ProphecyMap_k(s)\mid s\in S\}) \).
Finally, we define \( \FailureMap_k \).
For each thread identifier \(  i \), \(  \FailureMap_k(i) \defeq \TRUE \) iff the final failure in the suffix
\( C_k \stepC{D} \cdots \stepC{D} C_m \) is caused by thread \( i \) or by a descendant of \( i \) spawned after position
\( k \).
Otherwise, \( \FailureMap_k(i)= \FALSE \).

It is easy to check that \( \EConf_k^\pi  \stepC{D} \EConf_{k + 1}^\pi \) using the rule that corresponds to the rule used to derive \( C_k \stepC{D} C_{k + 1}\).
\end{proof}

\begin{theorem}[Contraposition of Soundness]\label{lem:contraposition-of-soundness}
  Let $(D, f(\seq{a}))$ be a program such that $\transmap{(D, f(\seq{a}))} = \CHCSystem$. If $(\SET{(f(\seq{a}))}, \emptyset) \mstepC{D} \FAIL$, then $\CHCSystem$ is unsatisfiable.
\end{theorem}
\begin{proof}
  Assume that $(\SET{(f(\seq{a}))}, \emptyset) \mstepC{D} \FAIL$. Then we obtain a corresponding execution trace in the extended semantics: $(\SET{(i, f(\seq{a}))}, \emptyset, \SET{i \mapsto 0}, \emptyset, \SET{i \mapsto \TRUE}) \mstepC{D} \FAIL$ by Lemma~\ref{lem:efail-construct}.
  The configuration $(\SET{(i, f(\seq{a}))}, \emptyset, \SET{i \mapsto 0}, \emptyset, \SET{i \mapsto \TRUE})$ is translated into $F(\TRUE, 0, \seq{a})$.
  By Lemmas \ref{lem:well-formedness-invariant}, \ref{lem:simulation-without-estep-fail} and \ref{lem:simulation-estep-fail}, $\Goal{F(\TRUE, 0, \seq{a})} \mSLDstep{\CHCSystem} \Goal{}$.
  Moreover, $\CHCSystem$ contains the CHC $\bot \Leftarrow F(\TRUE, t, \seq{a})$. Hence $\bot$ is derivable, and therefore $\CHCSystem$ is unsatisfiable.
\end{proof}

\subsection{Proof of Completeness}
The proof of the completeness is done by constructing a reduction sequence (for the extended configuration) that leads to an assertion failure from a proof of unsatisfiability.
Another invariant we shall use, apart from well-formedness, is that the configuration always includes a thread that is guessed to fail eventually.
\begin{definition}
  We say that an extended configuration $\EConf = ( \Threads, \EQu, \TSMap, \ProphecyMap, \FailureMap )$ \emph{may fail} if there exists a thread identifier $i$ that occurs in $\Threads$ such that $\FailureMap(i) = \TRUE$.
  \qed
\end{definition}

We first show that the proof of unsatisfiability can be ``scheduled'' according to the timestamps.
\begin{lemma}\label{lem:linearization-lemma}
  Let $D$ be a set of function definitions such that $\transmap {D} = \CHCSystem$.
  Suppose that \( \Goal{F(\TRUE, t_0, \seq{a})} \SLDstep{\CHCSystem}^* \Goal{} \), where $t_0$ is a constant.
  Then there exists goals $G_0, \ldots, G_m$ and an index $k$ such that
  \[
    G_0 = \Goal{F(\TRUE, t_0, \seq{a})}, \quad G_0 \SLDstep{\CHCSystem} G_1 \SLDstep{\CHCSystem} \cdots \SLDstep{\CHCSystem} G_m = \Goal{}
  \]
  and:
  \begin{enumerate}
    \item timestamps introduced in later steps are never smaller than those introduced in earlier steps. %
    \item $|G_j| \le |G_{j+1}|$ for each $j < k$, while $|G_j| > |G_{j+1}|$ for each $k \le j < m$.
  \end{enumerate}
\end{lemma}
\begin{proof}
  The sequence of resolution steps \( \Goal{F(\TRUE, t_0, \seq{a})} \SLDstep{\CHCSystem}^* \Goal{} \) can be reconstructed as a derivation tree where each resolution step can be written as an application of the following inference rule  %
\infrule[]
  {A_1 \subst \quad \ldots \quad A_n \subst}
  {\Head \subst}
  whose side condition is $\Head \Leftarrow A_1  \land \ldots \land A_n  \land \Constraint \in \CHCSystem$ and $\models \Constraint \subst$.
  Written this way, it is easy to see  that the timestamps introduced along each branch of the derivation tree are nondecreasing because, for each clause
  \[
    F_i(b, t, \seq x) \Leftarrow F_{i1} (b_{i1},t_{i1}, \seq{x}_{i1}) \land \cdots \land F_{i n_i} (b_{in_i},t_{in_i}, \seq{x}_{in_i})\land \Constraint
  \]
  in \( \transmap{D} \) the timestamps \(t_{ij} \) are never smaller than the timestamps \( t \).
  We first observe that any resolution step using a clause whose body contains no predicate atoms can be postponed until all steps using clauses with predicate
  atoms have been performed.
  While we apply clauses with a predicate in the body, we may choose the nodes in order of increasing introduced timestamps, with ties broken arbitrarily.
  This yields the desired sequence of resolution steps.
\end{proof}

\begin{lemma}[Simulation of SLD resolution]
  \label{lem:simulation-of-sld-step-without-fail}
  Let $D$ be a set of function definitions such that $\transmap {D} = \CHCSystem$. Let $G$ and $G'$ be goals such that $G \SLDstep{\CHCSystem} G'$.
  Suppose that $G$ and $G'$ occur consecutively in the sequence of resolution steps $G_0 \mSLDstep{\CHCSystem} G_k$ considered in \autoref{lem:linearization-lemma}. Let $\EConf_1$ be a well-formed extended configuration such that $\transmap {\EConf_1} = G$ and $\EConf_1$ may fail.

  Then there exists an extended configuration $\EConf_2$ such that $\transmap {\EConf_2} =  G'$, $\EConf_2$ may fail, and $\EConf_1 \stepC{D} \EConf_2$.
\end{lemma}
\begin{proof}
  Let $\EConf_1 = ( \Threads, \EQu, \TSMap, \ProphecyMap, \FailureMap  )$.

  We perform a case analysis according to from which statement the CHC used in the resolution step was transformed.
\begin{itemize}
  \item Case \texttt{if then else}.

    In this case, the resolution step may use one of the following two CHCs:
    \[
      F(b, t, \seq{x}) \Leftarrow a_0 \ne 0 \land H_1(b, t, \seq{a}_1)
    \] and
    \[
      F(b, t, \seq{x}) \Leftarrow a_0 = 0 \land H_2(b, t, \seq{a}_2)
    \] 
    where $f(\seq{x}:\seq{\bt}) = \ifstmt{a_0}{h_1(\seq{a}_1)}{h_2( \seq{a}_2 )}$.

    Since $G \SLDstep{\CHCSystem} G'$, there exists a thread $(i, f(\seq{a})) \in \Threads$ such that $G$ contains the atom translated from $( i, f(\seq{a}) )$, and the constraints of one of the above CHCs are satisfied. We distinguish the two cases according to which CHC is used.
    
    In the former case, we have $[\seq{a}/\seq{x}]a_0 \ne 0$. By the semantics of the target language, we obtain $[\seq{a}/\seq{x}]a_0 \eval n$ and $n \ne 0$. Hence, all the premises of \rn{EStep-If-True} are satisfied. Therefore, $\EConf_1 \stepC{D} \EConf_2$ by \rn{EStep-If-True} where $\EConf_2$ is obtained from $\EConf_1$ by replacing the thread $(i, f(\seq{a}))$ with $(i, [\seq{a}/\seq{x}]h_1(\seq{a}_1))$ in $\Threads$. Since $[\seq{a}/\seq{x}]h_1(\seq{a}_1)$ is translated into $H_1(b, t, [\seq{a}/\seq{x}]\seq{a}_1)$, $\EConf_2$ is translated into $G'$.
    This transition preserves the may-fail property because it does not modify the \( (\EQu,  \TSMap ,  \ProphecyMap, \FailureMap) \) part of the configuration.

    In the latter case, we have $[\seq{a}/\seq{x}]a_0 = 0$. This case is proved in the same way, using \rn{EStep-If-False} instead of \rn{EStep-If-True}.

  \item Case \texttt{spawn}.

    In this case, the resolution step may use the CHC
    \[
    F(b, t, \seq{x}) \Leftarrow H_1(b_1, t, \seq{a}_1) \land H_2(b_2, t, \seq{a}_2) \land b = b_1 \lor b_2
    \]
    where $f(\seq{x}:\seq{\bt}) = \spawnstmt{ h_1( \seq{a}_1 ) }{ h_2( \seq{a}_2 ) }$.

    Since $G \SLDstep{\CHCSystem} G'$, there exists a thread $(i, f( \seq{a} )) \in \Threads$ such that $G$ contains the atom translated from $( i, f( \seq{a} ) )$, and a substitution $\subst$ that grounds the above CHC and satisfies its constraints. Therefore, we obtain $\FailureMap(i) = \subst(b_1) \lor \subst(b_2)$.

    Let 
    \[\EConf_2 = ( (\Threads' \cup \SET{ (i, [\seq{a}/\seq{x}]h_1( \seq{a}_1 )), (j, [\seq{a}/\seq{x}]h_2( \seq{a}_2 ) }, \EQu, \TSMap [ j \mapsto \TSMap(i) ], \ProphecyMap, \FailureMap[i \mapsto \subst(b_1), j \mapsto \subst(b_2)])\]
    where $j$ is a fresh thread identifier and $\Threads' = \Threads \setminus \SET{ (i, f(\seq{a})) }$.
    Then $\EConf_1 \stepC{D} \EConf_2$. Moreover, $\EConf_2$ is translated into $G'$ by \rn{EStep-Spawn}.
    The transition $\EConf_1 \stepC{D} \EConf_2$ preserves the may-fail property because if $\FailureMap(i) = \TRUE$, then either $\FailureMap(i) = \subst(b_1) = \TRUE$ or $\FailureMap(j) = \subst(b_2) = \TRUE$.

  \item Case \texttt{new}.

    In this case, the resolution step may use the CHC
    \[
    F(b, t, \seq{x}) \Leftarrow H(b, t, \seq{a}') \land y =  x \land \SORTED(x)
    \]
    where $f( \seq{x}: \seq{\bt} ) = \newstmt{ s }{ r }{ h(b, t, \seq{a}') }$.

    Since $G \SLDstep{\CHCSystem} G'$, there exists a thread $(i, f( \seq{a} )) \in \Threads$ such that $G$ contains the atom translated from $( i, f( \seq{a} ) )$ and a substitution $\subst$ that grounds the above CHC and satisfies its constraints. Therefore, we obtain
    \begin{equation}
      \subst(r) =  \subst(s) \text{ and } \SORTED( \subst(s) ). \label{eq:lem:simulation-of-sld-step-without-fail:new}
    \end{equation}
    Let
    \[
      \EConf_2 = ( \Threads' \cup \SET{ (i, [s'/s, r'/r][\seq{a}/\seq{x}]h(\seq{a}')) }, \EQu[r' \mapsto (\SET{ s' }, \NIL, -\infty)], \TSMap, \ProphecyMap[s' \mapsto \subst(s), r' \mapsto \subst(r)], \FailureMap )
    \]
    where $\Threads' = \Threads \setminus \SET{ (i, f(\seq{a}')) }$. Then $\EConf_1 \stepC{D} \EConf_2$ by \rn{EStep-New} together with \eqref{eq:lem:simulation-of-sld-step-without-fail:new}, and $\EConf_2$ is translated into $G'$.
    The may-fail property is also preserved as \( \ProphecyMap \) is not modified in this step.

  \item Case \texttt{send}.

    In this case, the resolution step may use the CHC
    \[
    F(b, t, \seq{x}) \Leftarrow H(b, t', [s''/s]\seq{a}') \land s = (t', a_0) :: s'' \land t \le t'
    \]
    where $f(\seq{x}: \seq{\bt}) = \sendstmt{ a_0 }{ s }{ h(\seq{a}') }$.
    Since $G \SLDstep{\CHCSystem} G'$, there exists a thread $(i, f(\seq{a})) \in \Threads$ such that $G$ contains the atom translated from $( i, f( \seq{a} ) )$ and a substitution $\subst$ that grounds the above CHC and satisfies its constraints.
    Therefore, we obtain
    \begin{align}
      &\TSMap(i) \le \subst(t') \label{eq:lem:simulation-of-sld-step-without-fail:send:time} \\
      &\ProphecyMap(s') = (\subst(t'), [\seq{a}/\seq{x}]a_0) :: \subst(s'') \label{eq:lem:simulation-of-sld-step-without-fail:send:prophecy}
    \end{align}
    where $s' = [\seq{a}/\seq{x}]s$.
    
    Since $\EConf_1$ is well-formed, there exists a unique receiver $r$ such that $\EQu(r) = (S, \Eqc, t_S)$ and $s' \in S$. Write $S = \SET{s', s_1, \ldots, s_n}$.

    Let 
    \[
      \EConf_2 = ( \Threads' \cup \SET{ (i, [\seq{a}/\seq{x}]h(\seq{a}_1)) }, \EQu[ r \mapsto (S, \Eqc \cdot [(\subst(t'), n)], \subst(t')) ], \TSMap[i \mapsto \subst(t')], \ProphecyMap[s' \mapsto \subst(s'')], \FailureMap )
    \]
    where $\Threads' = \Threads \setminus \SET{ (i, f(\seq{a})) }$ and $n = [\seq{a}/\seq{x}]a_0$.
    By the third condition of the well-formedness,
    \begin{align}
      &t_S < \MinTimestamp{\ProphecyMap(s')}  \label{eq:lem:simulation-of-sld-step-without-fail:send:time:iii} \\
      &\SORTED(\ProphecyMap(s'))
    \end{align}
    Using \eqref{eq:lem:simulation-of-sld-step-without-fail:send:prophecy} and \eqref{eq:lem:simulation-of-sld-step-without-fail:send:time:iii}, we obtain
    \begin{align}
      t_S < \subst(t') \label{eq:lem:simulation-of-sld-step-without-fail:send:time:iv}
    \end{align}
    By the first condition of \autoref{lem:linearization-lemma}, together with the uniqueness of timestamps within each channel guaranteed by the Sorted condition, we obtain
    \begin{align*}
      &t' < \MinTimestamp{\sigma(s'')} \text{ and} \\
      &t' < \MinTimestamp{M(s_i)}  \text{ for each $i = 1, \ldots, n$}
    \end{align*}
    These together with \eqref{eq:lem:simulation-of-sld-step-without-fail:send:time} and \eqref{eq:lem:simulation-of-sld-step-without-fail:send:time:iv} imply $\EConf_1 \stepC{D} \EConf_2$ by \rn{EStep-Send} and $\EConf_2$ is translated into $G'$.
    The may-fail property is trivially preserved.
  \item Case \texttt{recv}.

    In this case, the resolution step may use the CHC
    \[
    F(b, t, \seq{x}) \Leftarrow H(b, t'', [r''/r]\seq{a}') \land r = (t', v)::r'' \land t' < t'' \land t \le t''
    \]
    where $f(\seq{x}:\seq{\bt}) = \recvstmt{ v }{ r  }{ h(\seq{a}') }$. Since $G \SLDstep{\CHCSystem} G'$, there exists a thread $(i, f(\seq{a})) \in \Threads$ such that $G$ contains the atom translated from $( i, f(\seq{a}) )$ and a substitution $\subst$ that grounds the above CHC and satisfies its constraints.
    Therefore, we obtain
    \begin{align}
      &\ProphecyMap(r') = (\subst(t'), \subst(v)) :: \subst(r'') \nonumber \\
      &\subst(t') < \subst(t'') \label{eq:lem:simulation-of-sld-step-without-fail:recv:time} \\
      &\TSMap(i) \le \subst(t'') \label{eq:lem:simulation-of-sld-step-without-fail:recv:time:ii}
    \end{align}
      where $r' = [\seq{a}/\seq{x}]a_0$.

    Since $\EConf_1$ is well-formed, $r' \in \DOM{\EQu}$. Write $\EQu(r') = (S, \Eqc, t_S)$ and $S = \SET{s_1, \ldots, s_n}$. 
    By the third condition of well-formedness, we have the following:
    \[
    (\subst(t'), \subst(v)) :: \subst(r'') =  \MergeT( \Eqc, \ProphecyMap(s_1), \ldots, \ProphecyMap(s_n) )
    \]
    Assume, for contradiction, that $\Eqc = \NIL$. Then $t'$ would be introduced later than $t''$. However, since $\subst(t') < \subst(t'')$, this contradicts the first condition of \autoref{lem:linearization-lemma}. Therefore, we have $\Eqc = (\subst(t'), \subst(v)) :: \Eqc'$.

    Let
    \[
      \EConf_2 = ( \Threads' \cup \SET{ (i, [n/v][\seq{a}/\seq{x}]h(\seq{a}')) }, \EQu[r' \mapsto (S, \Eqc', t_S)], \TSMap[i \mapsto \subst(t'')], \ProphecyMap[r' \mapsto \subst(r'')], \FailureMap )
    \]
    where $\Threads' = \Threads \setminus \SET{ (i, f(\seq{a})) }$ and $n = \subst(v)$.

    Then $ \EConf_1 \stepC{D} \EConf_2 $ by \rn{EStep-Recv} because the premises on timing are satisfied by \eqref{eq:lem:simulation-of-sld-step-without-fail:recv:time} and \eqref{eq:lem:simulation-of-sld-step-without-fail:recv:time:ii}.
    Moreover, $\EConf_2$ is translated into $G'$, and the may-fail property is trivially preserved.

  \item Case \texttt{clone}.

    In this case, the resolution step may use the CHC
    \[
    F(b, t, \seq{x}) \Leftarrow H(b, t, [z/s]\seq{a}') \land \MERGE(z, y, s)
    \]
    where $f(\seq{x}:\seq{\bt}) = \dupstmt{y}{s}{h(\seq{a}')}$.
    Since $G \SLDstep{\CHCSystem} G'$, there exists a thread $(i, f(\seq{a})) \in \Threads$ such that $G$ contains the atom translated from $( i, f(\seq{a}) )$ and a substitution $\subst$ that grounds the above CHC and satisfies its constraints.
    Therefore, we obtain $\MERGE(\subst(z), \subst(y), \ProphecyMap(s'))$ where $s' = [\seq{a}/\seq{x}]s$.

    Since $\EConf_1$ is well-formed, there exists a unique receiver $r$ such that $\EQu(r) = (S, \Eqc, t_S)$ and $s' \in S$. Write $S = \SET{s', s_1, \ldots, s_n}$.

    Let
    \[
      \EConf_2 = ( \Threads' \cup \SET{ (i, [y'/y][\seq{a}/\seq{x}]h(\seq{a}')) }, \EQu[r \mapsto (S \cup \SET{ y' }, \Eqc, t_S)], \TSMap, \ProphecyMap[s' \mapsto \subst(z), y' \mapsto \subst(y)], \FailureMap )
    \]
    where $\Threads' = \Threads \setminus \SET{ (i, f(\seq{a})) }$ and $y'$ is a fresh sender variable.

    Then $\EConf_1 \stepC{D} \EConf_2$ by \rn{EStep-Dup}, $\EConf_2$ is translated into $G'$ and the may-fail property is preserved.
   \end{itemize}
\end{proof}

\begin{lemma}\label{lem:simulation-of-sld-step-with-fail}
  Let $D$ be a set of function definitions such that $\transmap{D} =  \CHCSystem$. Let $G_k$ be the goal such that $G_k \mSLDstep{\CHCSystem} \Goal{}$ in \autoref{lem:linearization-lemma}, i.e., the goal after which the resolution steps only decrease the size of the goal. Let $\EConf$ be a well-formed extended configuration such that \( \transmap \EConf = G_k \) and $\EConf$ may fail.
  Then $\EConf \stepC{D} \FAIL$ by \rn{EStep-Fail}.
\end{lemma}
\begin{proof}
  Write $( \Threads, \EQu, \TSMap, \ProphecyMap, \FailureMap ) = \EConf$.
  By \autoref{lem:linearization-lemma}, each resolution step in $G_k \mSLDstep{\CHCSystem} \Goal{}$ strictly decreases the number of atoms in the goal. By the definition of \( \transmap{D}\), the only CHCs that can be used in such steps are the following:
  \[
    F(\TRUE, t, \seq{x}) \Leftarrow \seq{s} = \seq{\NIL}
  \]
  where $\seq{s} = \OutChannels{\seq{x}:\seq{\bt}}$ and $f( \seq{x}:\seq{\bt} ) = \FAIL \in D$, and
  \[
    F(\FALSE, t, \seq{x}) \Leftarrow \seq{s} = \seq{\NIL}
  \]
  where $\seq{s} = \OutChannels{\seq{x}:\seq{\bt}}$ and $f( \seq{x}:\seq{\bt} ) = M \in D$.

  We perform a case analysis on these CHCs.
\begin{itemize}
  \item Case $F(\TRUE, t, \seq{x}) \Leftarrow \seq{s} = \seq{\NIL}$

    In this case, there exists a thread identifier \( i \) such that $( i, f( \seq{a} ) ) \in \Threads$.
    By the constraints of this CHC, we obtain $\FailureMap(i) = \TRUE$ and $\ProphecyMap(s) = \NIL$ for each $s \in \OutChannels{\seq{ a }}$.
  \item Case $F(\FALSE, t, \seq{x}) \Leftarrow \seq{s} = \seq{\NIL}$

    In this case, there exists a thread identifier \( i \) such that $( i, f( \seq{a} ) ) \in \Threads$.
    By constraints of the CHC used in this case, we obtain $\FailureMap(i) = \FALSE$ and $\ProphecyMap(s) = \NIL$ for each $s \in \OutChannels{\seq{a}}$.
\end{itemize}
  By the above case analysis and the well-formedness of $\EConf$, we have $\ProphecyMap(s) = \NIL$ for each sender $s$ that occurs in $\Threads$. Moreover, for each $( i, f(\seq{a}) ) \in \Threads$, if $f(\seq{x}:\seq{\bt}) = \FAIL \in D$, then $\FailureMap(i) = \TRUE$ or $\FailureMap(i) = \FALSE$; otherwise, $\FailureMap(i) = \FALSE$. Since $\EConf$ may fail, there exists a thread identifier $i$ with $\FailureMap(i) = \TRUE$. Therefore, all the premises of \rn{EStep-Fail} are satisfied, and hence $\EConf \stepC{D} \FAIL$.
\end{proof}

\begin{theorem}[Contraposition of Completeness]
  Let $(D, f(\seq{a}))$ be a program such that $\transmap{(D, f(\seq{a}))} = \CHCSystem$. If $\CHCSystem$ is unsatisfiable, then $(\SET{(i, f(\seq{a}))}, \emptyset) \mstepC{D} \FAIL$.
\end{theorem}
\begin{proof}
  Assume that $\CHCSystem$ is unsatisfiable.
  By \autoref{lem:linearization-lemma}, there exist goals $G_0, \ldots, G_m$ and an index $k$ such that
  \[
    G_0 = \Goal{F(\TRUE, t_0, \seq{a})}, \quad G_0 \SLDstep{\CHCSystem} G_1 \SLDstep{\CHCSystem} \cdots \SLDstep{\CHCSystem} G_m = \Goal{}
  \]
  and satisfying the conditions in \autoref{lem:linearization-lemma}.

  Let $\EConf = (\SET{(i, f(\seq{a}))}, \emptyset, \SET{i \mapsto t_0}, \emptyset, \SET{i \mapsto \TRUE})$. Then \( \transmap \EConf = G_0 \). Since $\EConf$ is well-formed and may fail, by Lemmas \ref{lem:well-formedness-invariant}, \ref{lem:simulation-of-sld-step-without-fail} and \ref{lem:simulation-of-sld-step-with-fail}, we have $\EConf \mstepC{D} \FAIL$. By removing the additional information attached to extended configurations, we obtain $(\SET{(i, f(\seq{a}))}, \emptyset) \mstepC{D} \FAIL$.
\end{proof}

%% file: sections/EvaluationAppx.tex
\section{Benchmark Programs for the Experiments in \autoref{sec:implementation-evaluation}}
\label{appx:evaluation-examples}

\begin{example}
\label{ex:running-ex-e}
  The program of the benchmark \texttt{running-ex-e} is shown in \autoref{fig:multi-thread-send-e}. In this program, the loop that spawns sender threads iterates one more time than in \autoref{ex:running-ex}, which causes the assertion to fail when the number of messages received exceeds \(n\).
\begin{figure}[tbp]
\begin{verbatim}
fn msg_count(n: usize) {
    let (s, r) = channel();
    for i in 0..=n { // spawn n senders
        let s_clone = s.clone();
        thread::spawn(move || {
            s_clone.send(1).unwrap(); // each thread sends 1 to the channel
        });
    }
    let mut c = 0; // number of messages received
    loop {  // repeatedly receive and count messages
        let _v = r.recv().unwrap();
        c += 1;  // increments the message counter
        assert!(c <= n); // error if more than n messages are received
    }
}
\end{verbatim}
\caption{Program in \autoref{ex:running-ex} with bugs.}
\label{fig:multi-thread-send-e}
\end{figure}
\end{example}

\begin{example}
\label{ex:ack-error}
  The program of the benchmark \texttt{ack-e} is shown in \autoref{fig:ack-error}. In this program, the order of the values sent is reversed compared with \autoref{ex:dependency-threads}.
\begin{figure}[tbp]
\begin{verbatim}
fn main_with_bugs() {
  let (tx1, rx) = channel();
  let (ack_s, ack_r) = channel();
  let tx2 = tx1.clone();
  spawn(move || {
    tx1.send(2).unwrap();
  });
  spawn(move || {
    let x = ack_r.recv().unwrap();
    tx2.send(1).unwrap();
  });
  let v1 = rx.recv().unwrap();
  ack_s.send(1).unwrap();
  let v2 = rx.recv().unwrap();
  assert!(v1 == 1 && v2 == 2);
}
\end{verbatim}
  \caption{Program in \autoref{ex:dependency-threads} with bugs.}
\label{fig:ack-error}
\end{figure}
\end{example}

\begin{example}
\label{ex:multi-sends}
  The programs of the benchmarks \texttt{multi-sends} and \texttt{multi-sends-e} are shown in \autoref{fig:multi-sends}.
\begin{figure}[tbp]
\begin{multicols}{2}
\begin{verbatim}
fn main(n: usize) {
    let (s, r) = channel();
    spawn(move || {
        for i in 0..n {
            s.send(1).unwrap();
        }
    });
    let mut sum = 0;
    for i in 0..n {
        let v = r.recv().unwrap();
        sum += v;
    }
    assert!(sum == n);
}
\end{verbatim}
\begin{verbatim}
fn main_with_bugs(n: usize) {
    let (s, r) = channel();
    spawn(move || {
        for i in 1..n {
            s.send(1).unwrap();
        }
    });
    let mut sum = 0;
    for i in 1..n {
        let v = r.recv().unwrap();
        sum += v;
    }
    assert!(sum == n);
}
\end{verbatim}
\end{multicols}
\caption{Multiple sends program.}
\label{fig:multi-sends}
\end{figure}
  In the function \texttt{main}, one thread sends the value $1$ $n$ times, while another thread receives it $n$ times and asserts that their sum equals $n$. In the function \texttt{main\_with\_bugs}, although it should send and receive values $n$ times, it mistakenly does so only $n-1$ times.
\end{example}

\begin{example}
\label{ex:client-server}
  The program of the benchmark \texttt{client-server-e} is shown in \autoref{fig:client-server}. In the function \texttt{main\_with\_bugs}, although it should send values $n$ times, it mistakenly does so only $n-1$ times.
\begin{figure}[tbp]
\begin{verbatim}
fn main_with_bugs(n: usize) {
    let (s1, r1) = channel();
    let (s2, r2) = channel();
    spawn(move || {
        loop {
            let v = r2.recv().unwrap();
            s1.send(v + 1).unwrap();
        }
    });
    let mut sum = 0;
    for i in 1..n {
        s2.send(1).unwrap();
        let v = r1.recv().unwrap();
        sum += v;
    }
    assert!(sum == n * 2);
}
\end{verbatim}
\caption{Client-server program with bugs.}
\label{fig:client-server-bug}
\end{figure}
\end{example}

\begin{example}
\label{ex:functional-correctness}
  The programs of the benchmarks \texttt{calc-server} and \texttt{calc-server-e} are shown in \autoref{fig:functional-correctness} and \autoref{fig:functional-correctness-fail}. In the function \texttt{main}, a spawned thread acts as a server that performs addition when it receives $0$ as a command and subtraction when it receives $1$. The main thread sends $cmd$, $x$, and $y$ to the server in this order and asserts that the returned result is correct. In the function \texttt{main\_with\_bugs}, commands $0$ and $1$ are erroneously swapped.
\begin{figure}[tbp]
\begin{verbatim}
fn main(cmd: isize, x: isize, y: isize) {
    let (cmd_s, cmd_r) = channel();
    let (result_s, result_r) = channel();
    spawn(move || server(cmd_r, result_s));

    cmd_s.send(cmd).unwrap();
    cmd_s.send(x).unwrap();
    cmd_s.send(y).unwrap();  

    let result = result_r.recv().unwrap();
    if cmd == 0 {
        assert!(result == x + y);
    } else if cmd == 1 {
        assert!(result == x - y);
    }
    else {
        assert!(result == 0);
    }
}

fn server(cmd_r: Receiver<isize>, result_s: Sender<isize>) {
    loop {
        let cmd = cmd_r.recv().unwrap();
        let x = cmd_r.recv().unwrap();
        let y = cmd_r.recv().unwrap();
        if cmd == 0 {
            result_s.send(x + y).unwrap();
        } else if cmd == 1 {
            result_s.send(x - y).unwrap();
        } else {
            result_s.send(0).unwrap();
        }
    }
}
\end{verbatim}
\caption{More complex client-server program.}
\label{fig:functional-correctness}
\end{figure}

\begin{figure}[tbp]
\begin{verbatim}
fn main_with_bugs(cmd: isize, x: isize, y: isize) {
    let (cmd_s, cmd_r) = channel();
    let (result_s, result_r) = channel();
    spawn(move || server(cmd_r, result_s));

    cmd_s.send(cmd).unwrap();
    cmd_s.send(x).unwrap();
    cmd_s.send(y).unwrap();  

    let result = result_r.recv().unwrap();
    if cmd == 0 {
        assert!(result == x + y);
    } else if cmd == 1 {
        assert!(result == x - y);
    }
    else {
        assert!(result == 0);
    }
}

fn server(cmd_r: Receiver<isize>, result_s: Sender<isize>) {
    loop {
        let cmd = cmd_r.recv().unwrap();
        let x = cmd_r.recv().unwrap();
        let y = cmd_r.recv().unwrap();
        if cmd == 0 {
            result_s.send(x - y).unwrap();
        } else if cmd == 1 {
            result_s.send(x + y).unwrap();
        } else {
            result_s.send(0).unwrap();
        }
    }
}
\end{verbatim}
\caption{More complex client-server program with bugs.}
\label{fig:functional-correctness-fail}
\end{figure}
\end{example}

%% file: paper.bbl
\begin{thebibliography}{10}

\bibitem{abadi1991existence}
Mart{\'\i}n Abadi and Leslie Lamport.
\newblock The existence of refinement mappings.
\newblock {\em Theoretical Computer Science}, 82(2):253--284, 1991.

\bibitem{DBLP:conf/lics/AbdullaJ93}
Parosh~Aziz Abdulla and Bengt Jonsson.
\newblock Verifying programs with unreliable channels.
\newblock In {\em Proceedings of the Eighth Annual Symposium on Logic in
  Computer Science {(LICS} '93), Montreal, Canada, June 19-23, 1993}, pages
  160--170. {IEEE} Computer Society, 1993.
\newblock \href {https://doi.org/10.1109/LICS.1993.287591}
  {\path{doi:10.1109/LICS.1993.287591}}.

\bibitem{PrinciplesMC}
Christel Baier and Joost-Pieter Katoen.
\newblock {\em Principles of Model Checking}.
\newblock The MIT Press, 2008.

\bibitem{DBLP:conf/birthday/BeauxisPV08}
Romain Beauxis, Catuscia Palamidessi, and Frank~D. Valencia.
\newblock On the asynchronous nature of the asynchronous pi-calculus.
\newblock In Pierpaolo Degano, Rocco~De Nicola, and Jos{\'{e}} Meseguer,
  editors, {\em Concurrency, Graphs and Models, Essays Dedicated to Ugo
  Montanari on the Occasion of His 65th Birthday}, Lecture Notes in Computer
  Science, pages 473--492. Springer, 2008.
\newblock \href {https://doi.org/10.1007/978-3-540-68679-8\_29}
  {\path{doi:10.1007/978-3-540-68679-8\_29}}.

\bibitem{Bjorner15}
Nikolaj Bj{\o}rner, Arie Gurfinkel, Kenneth~L. McMillan, and Andrey
  Rybalchenko.
\newblock {Horn} clause solvers for program verification.
\newblock In {\em Fields of Logic and Computation {II} - Essays Dedicated to
  Yuri Gurevich on the Occasion of His 75th Birthday}, volume 9300 of {\em
  LNCS}, pages 24--51. Springer, 2015.
\newblock \href {https://doi.org/10.1007/978-3-319-23534-9_2}
  {\path{doi:10.1007/978-3-319-23534-9_2}}.

\bibitem{DBLP:journals/pacmpl/BurnOR18}
Toby~Cathcart Burn, C.{-}H.~Luke Ong, and Steven~J. Ramsay.
\newblock Higher-order constrained {Horn} clauses for verification.
\newblock {\em Proc. {ACM} Program. Lang.}, 2({POPL}):11:1--11:28, 2018.
\newblock \href {https://doi.org/10.1145/3158099} {\path{doi:10.1145/3158099}}.

\bibitem{DBLP:conf/aplas/Champion0S18}
Adrien Champion, Naoki Kobayashi, and Ryosuke Sato.
\newblock {HoIce}: An {ICE}-based non-linear {Horn} clause solver.
\newblock In Sukyoung Ryu, editor, {\em Programming Languages and Systems -
  16th Asian Symposium, {APLAS} 2018, Wellington, New Zealand, December 2-6,
  2018, Proceedings}, volume 11275 of {\em Lecture Notes in Computer Science},
  pages 146--156. Springer, 2018.
\newblock \href {https://doi.org/10.1007/978-3-030-02768-1\_8}
  {\path{doi:10.1007/978-3-030-02768-1\_8}}.

\bibitem{ClarkeModelChecking}
Edmund~M. Clarke, Orna Grumberg, and Doron~A. Peled.
\newblock {\em Model Checking}.
\newblock The MIT Press, 1999.

\bibitem{HandbookMC}
Edmund~M. Clarke, Thomas~A. Henzinger, Helmut Veith, and Roderick Bloem,
  editors.
\newblock {\em Handbook of Model Checking}.
\newblock Springer, 2018.

\bibitem{Angelis}
Emanuele {De Angelis}, Fabio Fioravanti, John~P. Gallagher, Manuel~V.
  Hermenegildo, Alberto Pettorossi, and Maurizio Proietti.
\newblock Analysis and transformation of constrained {Horn} clauses for program
  verification.
\newblock {\em Theory and Practice of Logic Programming}, 22(6):974–1042,
  2022.
\newblock \href {https://doi.org/10.1017/S1471068421000211}
  {\path{doi:10.1017/S1471068421000211}}.

\bibitem{denis2022creusot}
Xavier Denis, Jacques-Henri Jourdan, and Claude March{\'e}.
\newblock Creusot: A foundry for the deductive verification of {Rust} programs.
\newblock In {\em International Conference on Formal Engineering Methods},
  pages 90--105. Springer, 2022.

\bibitem{DBLP:conf/esop/FlanaganFQ02}
Cormac Flanagan, Stephen~N. Freund, and Shaz Qadeer.
\newblock Thread-modular verification for shared-memory programs.
\newblock In Daniel~Le M{\'{e}}tayer, editor, {\em Programming Languages and
  Systems, 11th European Symposium on Programming, {ESOP} 2002, held as Part of
  the Joint European Conference on Theory and Practice of Software, {ETAPS}
  2002, Grenoble, France, April 8-12, 2002, Proceedings}, Lecture Notes in
  Computer Science, pages 262--277. Springer, 2002.
\newblock \href {https://doi.org/10.1007/3-540-45927-8\_19}
  {\path{doi:10.1007/3-540-45927-8\_19}}.

\bibitem{travis2026verusbelt}
Travis Hance, Laila Elbeheiry, Yusuke Matsushita, and Derek Dreyer.
\newblock {VerusBelt}: A semantic foundation for {Verus}'s proof-oriented
  extensions to the {Rust} type system.
\newblock {\em Proc. ACM Program. Lang.}, 10(PLDI), June 2026.
\newblock \href {https://doi.org/10.1145/3808325} {\path{doi:10.1145/3808325}}.

\bibitem{DBLP:conf/cav/HenzingerJMQ03}
Thomas~A. Henzinger, Ranjit Jhala, Rupak Majumdar, and Shaz Qadeer.
\newblock Thread-modular abstraction refinement.
\newblock In Warren A.~Hunt Jr. and Fabio Somenzi, editors, {\em Computer Aided
  Verification, 15th International Conference, {CAV} 2003, Boulder, CO, USA,
  July 8-12, 2003, Proceedings}, Lecture Notes in Computer Science, pages
  262--274. Springer, 2003.
\newblock \href {https://doi.org/10.1007/978-3-540-45069-6\_27}
  {\path{doi:10.1007/978-3-540-45069-6\_27}}.

\bibitem{hojjat2018eldarica}
Hossein Hojjat and Philipp R{\"u}mmer.
\newblock The {ELDARICA} {Horn} solver.
\newblock In {\em 2018 Formal Methods in Computer Aided Design (FMCAD)}, pages
  1--7. IEEE, 2018.

\bibitem{huttel2016foundations}
Hans H{\"u}ttel, Ivan Lanese, Vasco~T Vasconcelos, Lu{\'\i}s Caires, Marco
  Carbone, Pierre-Malo Deni{\'e}lou, Dimitris Mostrous, Luca Padovani,
  Ant{\'o}nio Ravara, Emilio Tuosto, et~al.
\newblock Foundations of session types and behavioural contracts.
\newblock {\em ACM Computing Surveys (CSUR)}, 49(1):1--36, 2016.

\bibitem{jung2019future}
Ralf Jung, Rodolphe Lepigre, Gaurav Parthasarathy, Marianna Rapoport, Amin
  Timany, Derek Dreyer, and Bart Jacobs.
\newblock The future is ours: prophecy variables in separation logic.
\newblock {\em Proceedings of the ACM on Programming Languages}, 4(POPL):1--32,
  2019.

\bibitem{jung2015iris}
Ralf Jung, David Swasey, Filip Sieczkowski, Kasper Svendsen, Aaron Turon, Lars
  Birkedal, and Derek Dreyer.
\newblock Iris: Monoids and invariants as an orthogonal basis for concurrent
  reasoning.
\newblock In Sriram~K. Rajamani and David Walker, editors, {\em Proceedings of
  the 42nd Annual {ACM} {SIGPLAN-SIGACT} Symposium on Principles of Programming
  Languages, {POPL} 2015, Mumbai, India, January 15-17, 2015}, pages 637--650.
  {ACM}, 2015.
\newblock \href {https://doi.org/10.1145/2676726.2676980}
  {\path{doi:10.1145/2676726.2676980}}.

\bibitem{katsura2025automated}
Hiroyuki Katsura, Naoki Kobayashi, Ken Sakayori, and Ryosuke Sato.
\newblock Automated catamorphism synthesis for solving constrained {Horn}
  clauses over algebraic data types.
\newblock In {\em International Static Analysis Symposium}, pages 305--327.
  Springer, 2025.

\bibitem{kobayashi2003type}
Naoki Kobayashi.
\newblock Type systems for concurrent programs.
\newblock In {\em Formal Methods at the Crossroads. From Panacea to
  Foundational Support: 10th Anniversary Colloquium of UNU/IIST, the
  International Institute for Software Technology of The United Nations
  University, Lisbon, Portugal, March 18-20, 2002. Revised Papers}, pages
  439--453. Springer, 2003.

\bibitem{DBLP:conf/sas/0001NIU19}
Naoki Kobayashi, Takeshi Nishikawa, Atsushi Igarashi, and Hiroshi Unno.
\newblock Temporal verification of programs via first-order fixpoint logic.
\newblock In Bor{-}Yuh~Evan Chang, editor, {\em Static Analysis - 26th
  International Symposium, {SAS} 2019, Porto, Portugal, October 8-11, 2019,
  Proceedings}, volume 11822 of {\em Lecture Notes in Computer Science}, pages
  413--436. Springer, 2019.
\newblock \href {https://doi.org/10.1007/978-3-030-32304-2\_20}
  {\path{doi:10.1007/978-3-030-32304-2\_20}}.

\bibitem{DBLP:journals/pacmpl/KobayashiTST23}
Naoki Kobayashi, Kento Tanahashi, Ryosuke Sato, and Takeshi Tsukada.
\newblock {HFL(Z)} validity checking for automated program verification.
\newblock {\em Proc. {ACM} Program. Lang.}, 7({POPL}):154--184, 2023.
\newblock \href {https://doi.org/10.1145/3571199} {\path{doi:10.1145/3571199}}.

\bibitem{DBLP:conf/esop/0001TW18}
Naoki Kobayashi, Takeshi Tsukada, and Keiichi Watanabe.
\newblock Higher-order program verification via {HFL} model checking.
\newblock In Amal Ahmed, editor, {\em Programming Languages and Systems - 27th
  European Symposium on Programming, {ESOP} 2018, Held as Part of the European
  Joint Conferences on Theory and Practice of Software, {ETAPS} 2018,
  Thessaloniki, Greece, April 14-20, 2018, Proceedings}, Lecture Notes in
  Computer Science, pages 711--738. Springer, 2018.
\newblock \href {https://doi.org/10.1007/978-3-319-89884-1\_25}
  {\path{doi:10.1007/978-3-319-89884-1\_25}}.

\bibitem{DBLP:journals/fmsd/KomuravelliGC16}
Anvesh Komuravelli, Arie Gurfinkel, and Sagar Chaki.
\newblock {SMT}-based model checking for recursive programs.
\newblock {\em Formal Methods Syst. Des.}, 48(3):175--205, 2016.
\newblock URL: \url{https://doi.org/10.1007/s10703-016-0249-4}.

\bibitem{lattuada2023verus}
Andrea Lattuada, Travis Hance, Chanhee Cho, Matthias Brun, Isitha Subasinghe,
  Yi~Zhou, Jon Howell, Bryan Parno, and Chris Hawblitzel.
\newblock Verus: Verifying {Rust} programs using linear ghost types.
\newblock {\em Proc. {ACM} Program. Lang.}, 7({OOPSLA1}):286--315, 2023.
\newblock \href {https://doi.org/10.1145/3586037} {\path{doi:10.1145/3586037}}.

\bibitem{10.1007/978-3-540-74061-2_14}
Alexander Malkis, Andreas Podelski, and Andrey Rybalchenko.
\newblock Precise thread-modular verification.
\newblock In Hanne~Riis Nielson and Gilberto Fil{\'e}, editors, {\em Static
  Analysis}, pages 218--232, Berlin, Heidelberg, 2007. Springer Berlin
  Heidelberg.

\bibitem{matsushita2022rusthornbelt}
Yusuke Matsushita, Xavier Denis, Jacques{-}Henri Jourdan, and Derek Dreyer.
\newblock {RustHornBelt}: A semantic foundation for functional verification of
  {Rust} programs with unsafe code.
\newblock In Ranjit Jhala and Isil Dillig, editors, {\em {PLDI} '22: 43rd {ACM}
  {SIGPLAN} International Conference on Programming Language Design and
  Implementation, San Diego, CA, USA, June 13 - 17, 2022}, pages 841--856.
  {ACM}, 2022.
\newblock \href {https://doi.org/10.1145/3519939.3523704}
  {\path{doi:10.1145/3519939.3523704}}.

\bibitem{matsushita2025nola}
Yusuke Matsushita and Takeshi Tsukada.
\newblock Nola: Later-free ghost state for verifying termination in {Iris}.
\newblock {\em Proc. {ACM} Program. Lang.}, 9({PLDI}):98--124, 2025.
\newblock \href {https://doi.org/10.1145/3729250} {\path{doi:10.1145/3729250}}.

\bibitem{matsushita2021rusthorn}
Yusuke Matsushita, Takeshi Tsukada, and Naoki Kobayashi.
\newblock {RustHorn}: {CHC}-based verification for {Rust} programs.
\newblock {\em ACM Transactions on Programming Languages and Systems (TOPLAS)},
  43(4):1--54, 2021.

\bibitem{ogawa2025thrust}
Hiromi Ogawa, Taro Sekiyama, and Hiroshi Unno.
\newblock Thrust: {A} prophecy-based refinement type system for {Rust}.
\newblock {\em Proc. {ACM} Program. Lang.}, 9({PLDI}):2056--2080, 2025.
\newblock \href {https://doi.org/10.1145/3729333} {\path{doi:10.1145/3729333}}.

\bibitem{10.1007/978-3-032-22723-2_11}
Toby Ueno and Ankush Das.
\newblock Practical refinement session type inference.
\newblock In Robbert Krebbers, editor, {\em Programming Languages and Systems},
  pages 298--330, Cham, 2026. Springer Nature Switzerland.

\bibitem{DBLP:journals/pacmpl/UnnoTGK23}
Hiroshi Unno, Tachio Terauchi, Yu~Gu, and Eric Koskinen.
\newblock Modular primal-dual fixpoint logic solving for temporal verification.
\newblock {\em Proc. {ACM} Program. Lang.}, 7({POPL}):2111--2140, 2023.
\newblock \href {https://doi.org/10.1145/3571265} {\path{doi:10.1145/3571265}}.

\end{thebibliography}
